\documentclass[12pt]{article}
\usepackage[margin=1.25in]{geometry} 
\usepackage{amsmath,amsthm,amssymb,todonotes}
\usepackage{multirow}
\usepackage{times}
\usepackage{bm}
\usepackage{listings}
\usepackage{graphicx}
\usepackage{enumerate}
\usepackage[colorlinks,
linkcolor=blue,
anchorcolor=blue,
citecolor=blue
]{hyperref}
\usepackage{dsfont}
\usepackage{float}
\usepackage[title]{appendix}
\usepackage{booktabs}
\usepackage{caption}
\usepackage{tikz}
\usetikzlibrary{decorations.pathreplacing}
\usetikzlibrary{arrows,shapes}
\usetikzlibrary{arrows.meta,bending,positioning}
\usetikzlibrary{arrows.meta,decorations.pathreplacing}
\usepackage{subcaption}
\usepackage{graphicx}

\usepackage{pgfplots}
\usepgfplotslibrary{groupplots,fillbetween}
\usetikzlibrary{arrows.meta,positioning,calc}
\pgfplotsset{compat=1.18}

\usepackage{chngcntr}
\usepackage{leftindex}

\usepackage{natbib}

\usepackage[title]{appendix}

\usepackage{bbm}

\usepackage{setspace}
\theoremstyle{plain}

\newtheorem{proposition}{Proposition}
\newtheorem{theorem}{Theorem}
\newtheorem{corollary}{Corollary}

\newtheorem{lemma}{Lemma}

\newtheorem{definition}{Definition}

\newtheorem{remark}{Remark}

\makeatletter

\newcommand{\Rmnum}[1]{\expandafter\@slowromancap\romannumeral #1@}

\title{Resource Allocation with Midterm Review\footnote{I thank Yi Chen, Piotr Dworczak, Jeffrey Ely, Felix Zhiyu Feng, Andrew Ferdowsian, Thomas Gresik, Yingni Guo, Maciej Kotowski,  Alessandro Pavan, and Joseph Root for helpful comments.}}
\author{Yijun Liu\footnote{Department of Economics, University of Notre Dame, yliu2346@nd.edu.}}

\begin{document}
\maketitle

\begin{abstract}
    A principal funds an agent's project over two periods, with a midterm review providing a public signal about project quality. The agent privately knows the quality and always prefers more funding, while the principal wants to match funding to quality. We identify two reversal phenomena in the optimal mechanism. First, the principal's expected payoff is strictly higher in the first period than in the second, even though she is better informed after the review.  Second, higher-quality projects receive more funding initially but less funding following a bad review, and may receive less funding following a good review as well.
\end{abstract}

\newpage

\section{Introduction}

Many resource-allocation relationships unfold over time and include a scheduled review. For instance, a research funding agency awards a multi-year grant and reviews progress at the midpoint. A firm funds an R\&D project and decides on future funding after reviewing early results. A government funds a contractor-run program and adjusts subsequent funding after an evaluation. These settings share three features. First, the resource recipient knows more about the project's quality than the resource provider. Second, the provider would like to tailor resources to quality---allocating more to higher-quality projects---while the recipient always prefers more resources regardless of quality: a researcher prefers a larger grant, and a project manager a larger budget, whether or not the project merits it. Third, monetary transfers are typically not used to align interests: the agency does not sell grants, and the project lead is not paid for honest reporting.

Because the recipient is biased, the provider cannot simply rely on the recipient to truthfully reveal the project's quality. Instead, the provider often turns to an additional source of information, such as an expert evaluation or a progress review. Such information, however, is usually unavailable at the outset. A project generates evidence of its quality only once it is running, and it cannot run without initial support; the agency must fund the grant before there is a progress report to evaluate. The provider must therefore commit resources while uninformed, learn something in the interim, and then decide on future support.

The main contribution of this paper is to characterize how resources should be allocated when informative signals arrive between allocation decisions. We identify two reversal phenomena in the optimal mechanism. First, the principal's expected payoff is strictly higher in the first period than in the second, even though she is better informed after the review. Second, higher-quality projects receive less funding after a bad review, and may receive less funding after a good review as well. Together, these results show that interim information can reshape allocations both before and after its arrival in nontrivial ways. We characterize the precise intertemporal trade-offs underlying these results and provide economic intuition for both.

To study this problem, we formulate a two-period delegation model. In each period, the principal chooses an action $x_t$, interpreted as the resources (funding) allocated to the agent's project. The project has a fixed quality $\theta$, privately known to the agent. The principal has a quadratic loss and wants to match allocation to quality in each period, whereas the agent's payoff is the total amount of resources he receives, so he always prefers more. Between the two periods, a public binary signal is realized, either high or low, with higher-quality projects more likely to generate the high signal. At the outset, the principal commits to a first-period allocation that depends on the agent's report and a second-period allocation that depends on both the report and the realized signal.

Despite the simple setup, the optimal mechanism generally does not admit a closed-form solution. As discussed in Section~\ref{sec:mainresults}, the principal's problem can be formulated as an optimal control problem. Even the relaxed problem leads to a second-order differential equation whose coefficients depend on the type distribution and signal function.\footnote{The relaxed problem is obtained from the principal's original problem by omitting the participation constraint and the monotonicity constraint implied by incentive compatibility.} Nonetheless, we can still characterize the qualitative structure of the optimal mechanism through several general properties.

How should the principal use information from the midterm review? Although the principal is better informed in the second period, Theorem~\ref{thm:frontloading-information} shows that the optimal mechanism brings the midterm review's informational value forward to the first period, and does so to the extent that the principal's expected payoff is strictly higher in the first period than in the second. We refer to this phenomenon as \emph{frontloading information value}.

To understand this result, note that the agent cares only about total resources, so the principal can rebalance allocations across periods without changing reporting incentives. Optimality requires her to use this flexibility to equalize expected marginal losses across periods. Under quadratic loss, this implies that the first-period allocation equals the expected second-period allocation. We refer to this equality as the \emph{balancing condition}. Despite the same mean allocation across periods, the second-period allocation still fluctuates with the signal around this mean, creating additional variance and hence higher expected loss in the second period. We show that the frontloading result extends to a broader class of loss functions, as discussed in Section \ref{sec:frontload}.

When project quality follows a Beta distribution, including the uniform distribution as a special case, and the signal probability is linear in quality, the optimal mechanism admits a closed-form solution that provides a clear geometric illustration of frontloading information value (Proposition~\ref{prop:beta-prior}).

What does the optimal mechanism look like in general settings? Theorem~\ref{thm:general-monotonicity} shows that the first-period allocation $x_1^*(\theta)$ is increasing in project quality $\theta$: the initial allocation is responsive to quality, with higher-quality projects receiving more resources
(one manifestation of the frontloaded information value). 
Making the first-period allocation responsive, however, is not free. Incentive compatibility requires the second-period allocations to become negatively responsive to quality, in a probability-weighted sense. Theorem~\ref{thm:general-monotonicity} also shows that the second-period allocation after a low signal, $x_{2l}^*(\theta)$, is decreasing in $\theta$, so higher-quality projects receive fewer resources after a bad review. The monotonicity of the allocation following a high signal, $x_{2h}^*(\theta)$, depends on the type distribution and the signal function. It need not be monotone over the entire support, although it cannot be increasing for types sufficiently close to the highest type.

Two forces shape the second-period allocations. The first is an \emph{incentive effect}. Incentive compatibility requires the gap between $x_{2h}^*(\theta)$ and $x_{2l}^*(\theta)$ to increase with project quality $\theta$. This widening gap pushes $x_{2h}^*(\theta)$ upward and $x_{2l}^*(\theta)$ downward.
The second is a \emph{surprise effect}, which rewards more for a more surprising good review and punishes more for a more surprising bad review.
A low signal is more surprising for higher-quality projects, so they face a larger punishment after a low signal. This force pushes $x_{2l}^*(\theta)$ downward. Conversely, a high signal is less surprising for higher-quality projects, so they receive a smaller reward after a high signal. This force also pushes $x_{2h}^*(\theta)$ downward. The two forces therefore work in the same direction for $x_{2l}^*(\theta)$, yielding a decreasing allocation after a low signal. For $x_{2h}^*(\theta)$, however, they work in opposite directions, so its monotonicity is uncertain and depends on the type distribution and the signal function.

 Nevertheless, we show that $x_{2h}^*(\theta)$ cannot be increasing for types near the top of the distribution. For the incentive effect to dominate the surprise effect in this region, the gap between $x_{2h}^*(\theta)$ and $x_{2l}^*(\theta)$ must widen sufficiently rapidly. Such rapid widening, however, is not optimal for sufficiently high types, because there are almost no higher types left whose allocations and incentives can benefit from a further increase in the gap. When the highest type is not guaranteed to receive a high signal, or when the allocation rules are Lipschitz continuous, $x_{2h}^*(\theta)$ is strictly decreasing for types near the top.

Thus, the analogue of efficiency at the top in our setting is not to assign the highest type its first-best allocation, but rather for the rate at which the gap between $x_{2h}^*(\theta)$ and $x_{2l}^*(\theta)$ widens to vanish as $\theta$ approaches the top of the support. As a result, the incentive effect vanishes near the top, leaving the surprise effect to dominate and making the allocation following a high signal decreasing for sufficiently high types. This is a distinctive feature of our setting.

Finally, Theorem~\ref{thm:sandwich} shows that the optimal mechanism also satisfies a \emph{sandwich property}: for every project quality, the initial allocation lies between the two possible continuation allocations, i.e., $x_{2l}^*(\theta) \leq x_1^*(\theta) \leq x_{2h}^*(\theta).$ Thus, a good review increases resources relative to the initial allocation, whereas a bad review reduces them.
Combining this property with the monotonicity results gives a concrete picture of how the optimal mechanism works.
Projects that report higher quality receive more funding initially but face harsher consequences (lower funding) after a bad review. A good review increases funding, but the resulting allocation need not preserve the initial ranking of projects. For types near the top, they receive the highest initial allocations but also face strict continuation terms: a large reduction in funding following a bad review and a relatively small increase following a good review.

The rest of the paper is organized as follows. Section~\ref{sec:model} presents the model, characterizes incentive compatibility, and introduces the static optimal  benchmark. 
Section~\ref{sec:example} solves the problem in closed form for Beta priors with a linear signal probability and provides a geometric illustration of frontloading information value. Section~\ref{sec:mainresults} presents the main results: it establishes frontloading, monotonicity, and the sandwich property for general priors and signal functions, illustrates these results with numerical examples, and interprets the optimal mechanism. Section~\ref{sec:conclusion} concludes. All proofs are provided in the Appendix.

\subsection{Related Literature}
\label{sec:literature}

This paper contributes to the literature on delegation, including \cite{Holmstrom1980Theory}, \cite{MelumadandShibano1991}, \cite{Alonso2008}, \cite{KOESSLER20121850}, and \cite{AmadorBagwell2013}. The canonical delegation problem studies static decision making with partially aligned preferences. In contrast, we study a dynamic allocation problem with a fully biased agent.

Our paper is more closely related to the literature on dynamic delegation, including \cite{Guo2016}, \cite{frankel2016discounted}, \cite{KRAHMER2016849}, \cite{guo2021dynamic}, and \cite{chen2022dynamic}. Among these papers, we are most closely related to \cite{chen2022dynamic}, which also studies a dynamic delegation problem with a fully biased agent. In their setting, the state evolves over time and is observed only by the agent. The principal screens the agent by exploiting persistence in the state, rather than the public interim information as in our setting.

  Among the static delegation literature, \cite{kattwinkel2022mechanisms} and \cite{kattwinkel2023costless} also consider fully biased agents. In the former, screening is restored by exploiting correlation across agents; in the latter, costly verification provides the key screening instrument, while a costless correlated signal helps the principal economize on verification.
 \cite{halac2020commitment} likewise study static delegation with costly verification, but under partially aligned preferences. 
 Our setting instead has a single agent, no verification technology, and a public signal that arrives only after the first allocation is made; truthful reporting is supported through the joint design of initial and review-contingent continuation allocation.

\cite{boleslavsky2016evolving} study a related problem with two types and many periods. In our setting, project quality is a continuum, which lets us ask a question with no analogue in the two-type case: how allocation rules on either side of the review  respond to quality, which is the content of our monotonicity results.
\cite{lipnowski2020repeated} study a delegation decision making problem in a repeated allocation game without commitment.

This paper also relates to the literature on mechanisms with ex post verification, including work without transfers, such as \cite{GlazerRubinstein2004}, \cite{BenPorath&Dekel&Lipman2014}, \cite{Mylovanov&Zapechelnyuk2017}, and \cite{li2020mechanism}, as well as work with transfers, such as \cite{riordan1988optimal}, \cite{DeMarzo&Kremer&Skrzypacz2005}, \cite{Deb&Mishra2014}, and \cite{Ekmekci&Kos&Vohra2016}.  Unlike many of these papers, our signal provides only noisy information about the state rather than conclusive information. In addition, our setting is dynamic: the principal makes decisions both before and after the signal arrives.

Moreover, this paper is also related to the literature on mechanism design with evidence, including \cite{Green&Laffont1986}, \cite{BULL200775}, \cite{DENECKERE2008487}, \cite{BENPORATH20121689}, and \cite{KARTIK2012323}.  In these papers, the agent submits evidence or partially verifiable reports, with the underlying state constraining what the agent can substantiate. In our model, by contrast, the agent can make any report, while a noisy public review is generated exogenously between the two allocation decisions.

Finally, this paper is related to the literature on dynamic moral hazard with interim performance information. \cite{cetemen2026midtermreview} study how a midterm review can provide incentives for future effort, while \cite{varas2020random} study optimal dynamic monitoring when information serves both learning and incentive purposes. \cite{ely2025feedback} consider the joint design of incentives and performance feedback, and \cite{chen2013interim} analyze a setting in which initial effort affects an interim outcome, which in turn affects the technology of the final outcome.
\cite{madsen2025performance} study monetary and non-monetary performance incentives in a moral hazard setting; they identify a related interaction between incentive requirements and success probabilities: as target effort increases, punishments strengthen, while rewards can be nonmonotone. In these papers, interim information is used to incentivize hidden effort. In contrast, in our setting the agent privately knows a fixed project quality, and the anticipated review can be used to elicit this information at the outset.

\section{Model}
\label{sec:model}

A principal (she) and an agent (he) interact over two periods, $t\in\{1,2\}$. In each period, the principal chooses an action $x_t\in\mathbb{R}$, interpreted as the amount of resources allocated to the agent's project.
The project has quality $\theta$, which is privately known to the agent and remains fixed across the two periods. The quality is drawn from a commonly known distribution  
 $F$ with density $f$ and support $\Theta\equiv[\underline{\theta},\bar{\theta}]$, where
$0\le\underline\theta<\bar\theta$. We assume that $f$ is continuous and $f(\theta)>0$ for all $\theta\in\Theta$.

After the first period, a midterm review is conducted, generating a public signal $\sigma$ correlated with the project quality. The signal is binary, taking values in $\mathcal{S}\equiv\{H,L\}$, where $H$ denotes a high signal and $L$ a low signal. The signal is generated according to the following conditional probabilities:
\[
    \Pr(\sigma=H\mid\theta)=s(\theta),
    \qquad
    \Pr(\sigma=L\mid\theta)=1-s(\theta),
\]
where $s:\Theta\to[0,1]$ is continuously differentiable with $s'(\theta)>0$ for all $\theta\in\Theta$. That is, a higher-quality project is more likely to generate a high signal.

The principal has quadratic loss in each period and seeks to match her actions to the project quality.\footnote{There are many settings in which the principal cares about the allocation in each period, rather than only about the total amount. For instance, a funding agency may want to ensure that a project receives sufficient resources throughout its duration, rather than concentrating funding in the early or late stages. Similarly, a firm may prefer to maintain a steady level of investment in an R\&D project to support continuous progress.}
Her payoff in period $t$ is $-(x_t-\theta)^2$, so her first-best allocation is $x_t=\theta$. In contrast, the agent is fully biased: his payoff in each period is simply the action taken, $x_t$. He therefore always prefers a higher action.

The principal has commitment power and commits to a mechanism at the start of the game. By the revelation principle, we restrict attention to direct mechanisms. A mechanism $\mathcal{M}$ is a triple
$\{x_1,x_{2h},x_{2l}\}$,
where $x_1,x_{2h},x_{2l}:\Theta\to\mathbb{R}$. Given a report $\hat{\theta}$, $x_1(\hat{\theta})$ is the first-period action, while $x_{2h}(\hat{\theta})$ and $x_{2l}(\hat{\theta})$ are the second-period actions following signals $H$ and $L$, respectively.\footnote{Because the agent's payoff is linear in the actions and the principal's loss is quadratic, it is without loss of generality to restrict attention to deterministic mechanisms.} For ease of exposition, we restrict attention to continuous action rules, $x_1,x_{2h},x_{2l}\in C(\Theta)$.\footnote{Continuity is imposed only to avoid null-set qualifications and to state uniqueness and properties of the optimal mechanism pointwise. If instead one permits Borel-measurable action rules with finite expected quadratic loss, the same optimality conclusions continue to hold, while uniqueness and the properties of the optimal mechanism are identified only up to $F$-null sets.} After observing his type, the agent decides whether
to participate. If he declines participation, his payoff is zero.

A type-$\theta$ agent's expected payoff under the mechanism is
\begin{equation*}
    u(\theta)
    \equiv
    x_1(\theta)
    +s(\theta)\,x_{2h}(\theta)
    +\bigl(1-s(\theta)\bigr)x_{2l}(\theta).
\end{equation*}
The principal's expected payoff is
\begin{equation*}
\begin{aligned}
    U_p(x)
    \equiv
    -\int_{\Theta}
    \Big[
        &\bigl(x_1(\theta)-\theta\bigr)^2
        +s(\theta)
        \bigl(x_{2h}(\theta)-\theta\bigr)^2
        +
        \bigl(1-s(\theta)\bigr)
        \bigl(x_{2l}(\theta)-\theta\bigr)^2
    \Big]\,dF(\theta).
\end{aligned}
\end{equation*}

\subsection{Principal's optimization problem}

A direct mechanism is implementable if it is incentive compatible and individually rational: truthful reporting is optimal for every type, and every type prefers participation to opting out.
\begin{definition}
    A mechanism $\{x_1,x_{2h},x_{2l}\}$ is implementable if, for every $\theta \in \Theta$:
    \[
  \theta
    \in
    \arg\max_{\hat{\theta}\in\Theta}
    \left\{
        x_1(\hat{\theta})
        +s(\theta)x_{2h}(\hat{\theta})
        +\bigl(1-s(\theta)\bigr)x_{2l}(\hat{\theta})
    \right\}, \tag{IC}
\]
\[
    u(\theta) \geq 0. \tag{IR}
\]
    \end{definition}

Define the signal-contingent spread by $\Delta(\theta)
    \equiv
    x_{2h}(\theta)-x_{2l}(\theta)$. Incentive compatibility can then be characterized as follows.

    \begin{lemma}
\label{lem:IC}
    A mechanism $\{x_1,x_{2h},x_{2l}\}$ is incentive compatible if and only if the
following conditions hold for all $\theta\in\Theta$:
\begin{enumerate}[(i)]
    \item
    $\Delta(\theta)$ is increasing in $\theta$,\footnote{Throughout the paper, ``increasing'' and ``decreasing'' are in the weak sense unless explicitly qualified as ``strictly.''}

    \item $u(\theta)
        =
        u(\underline{\theta})
        +
        \int_{\underline{\theta}}^\theta
        s'(t)\Delta(t)\,dt.$
\end{enumerate}
Equivalently, condition {\rm (ii)} can be written as $ u'(\theta) = s'(\theta)\Delta(\theta)$
for every
$\theta\in\operatorname{int}(\Theta)$.
\end{lemma}

To interpret the lemma, the expected payoff of a type-$\theta$ agent who reports $\hat{\theta}$ is
\[
   x_1(\hat{\theta})+x_{2l}(\hat{\theta})
    +s(\theta)\Delta(\hat{\theta}).
\]
One can therefore view $s(\theta)$ as the agent's effective screening type and $\Delta$ as the allocation in a standard one-dimensional screening model. The agent's payoff is supermodular in $(s(\theta),\Delta)$. Since $s(\theta)$ is strictly increasing, incentive compatibility requires $\Delta$ to be increasing in $s(\theta)$, or equivalently in $\theta$. The second condition in Lemma~\ref{lem:IC} is the corresponding Envelope condition.

By Lemma~\ref{lem:IC}, the principal's problem can be
written as
\begin{align*}
    \max_{\substack{x_1,x_{2h},x_{2l}\in C(\Theta)}}
    \;&
    -\int_{\Theta}
    \Big[
        \bigl(x_1(\theta)-\theta\bigr)^2
        +
        s(\theta)
        \bigl(x_{2h}(\theta)-\theta\bigr)^2
        +
        \bigl(1-s(\theta)\bigr)
        \bigl(x_{2l}(\theta)-\theta\bigr)^2
    \Big]\,dF(\theta)
    \\[0.8em]
    \text{s.t.}\quad&
    \Delta(\theta)
    \text{ is increasing in $\theta$},
    \\
    &
     u(\theta)
    =
    u(\underline{\theta})
    +
    \int_{\underline{\theta}}^\theta
    s'(t)\Delta(t)\,dt, \\
    & u(\theta) \geq 0.
\end{align*}

Throughout the paper, we denote the optimal mechanism by $\{x^*_1(\theta), x^*_{2h}(\theta), x^*_{2l}(\theta)\}$, the agent's interim expected payoff under the optimal mechanism by $u^*(\theta)$, and the corresponding signal-contingent spread by $\Delta^*(\theta)$.

\subsection{The Static Optimal Benchmark}

How should the principal use the midterm review information? A natural approach is to use it only after it arrives, in the second period. Under this approach, the principal treats each period as an independent static problem with the second-period type distribution updated to the posterior given the signal realization. We refer to the resulting mechanism as the \emph{static optimal mechanism}, which serves as a benchmark for comparison with the optimal mechanism.

In a static setting, the principal cannot elicit information from a fully biased agent. Equivalently, the only incentive-compatible mechanism is a constant allocation rule. Under quadratic loss, the optimal constant allocation equals the project's expected quality. Formally, let $\{a_1, a_{2h}, a_{2l}\}$ denote the static optimal mechanism. Then $a_1$ is the prior mean, while $a_{2h}$ and $a_{2l}$ are the posterior means following signals $H$ and $L$, respectively.\footnote{Since project quality is nonnegative,
$a_1$, $a_{2h}$, and $a_{2l}$ are strictly positive. So the static optimal mechanism
satisfies both incentive compatibility and individual rationality.
} Their formulas are:
\[
    a_1
    \equiv
    \int_{\Theta}\theta\,dF(\theta), \quad a_{2h}
    \equiv
    \frac{
        \int_{\Theta}\theta s(\theta)\,dF(\theta)
    }{
         \int_{\Theta}s(\theta)\,dF(\theta)
    }, \quad   a_{2l}
        \equiv
        \frac{
            \int_{\Theta}\theta (1-s(\theta))\,dF(\theta)
        }{
            \int_{\Theta}(1-s(\theta))\,dF(\theta)
        }.
\]

In the following sections, we characterize the optimal mechanism, compare it with this static benchmark, and explain how it uses the midterm review information.

\section{Example}
\label{sec:example}

In this section, we consider a specific example with a uniform distribution and a linear signal function, for which the optimal mechanism admits a closed-form solution. We then show in Section~\ref{sec:mainresults} that, under general distributions and signal functions, the optimal mechanism shares several key features with this example.

Suppose the project quality $\theta$ is uniformly distributed on $[0,1]$, and the signal function is $s(\theta)=\theta$. Then the optimal mechanism is given by:
\[
    x_1^*(\theta)
    =
    \frac{2}{5}+\frac{1}{5}\theta,
    \qquad
    x_{2h}^*(\theta)
    =
    \frac{4}{5}-\frac{1}{5}\theta,
    \qquad
    x_{2l}^*(\theta)
    =
    \frac{2}{5}-\frac{1}{5}\theta.
\]

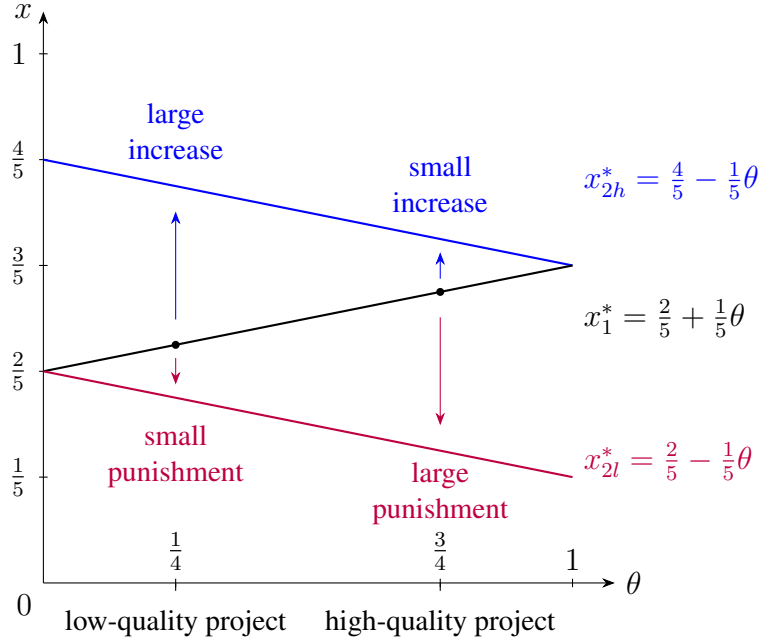
\begin{figure}[htbp]
    \centering
    \begin{tikzpicture}[
        x=7cm,
        y=7cm,
        >={Stealth[length=1.7mm]},
        adjustment/.style={->, thick},
        note/.style={font=\small, align=center}
    ]



        \draw[->]
            (0,0) -- (1.08,0) node[right] {$\theta$};
        \draw[->]
            (0,0) -- (0,1.08) node[left] {$x$};

        \node[below left] at (0,0) {$0$};

    \foreach \y/\lab in {
    0.2/\frac{1}{5},
    0.4/\frac{2}{5},
    0.6/\frac{3}{5},
    0.8/\frac{4}{5},
    1/1
} {
    \draw (-0.008,\y) -- (0.008,\y);
    \node[left=2pt] at (0,\y) {$\lab$};
}

        \foreach \t/\lab in {
            0.25/\frac14, 0.75/\frac34, 1/1
        } {
            \draw (\t,-0.008) -- (\t,0.008);
            \node[above=0.8pt] at (\t,0) {$\lab$};
        }


        \draw[blue, thick]
            (0,0.8) -- (1,0.6);

        \draw[black, thick]
            (0,0.4) -- (1,0.6);

        \draw[purple, thick]
            (0,0.4) -- (1,0.2);

        \node[blue, right] at (1,0.76)
            {$x^*_{2h}=\frac{4}{5}-\frac{1}{5}\theta$};

        \node[black, right] at (1,0.50)
            {$x^*_{1}=\frac{2}{5}+\frac{1}{5}\theta$};

        \node[purple, right] at (1,0.23)
            {$x^*_{2l}=\frac{2}{5}-\frac{1}{5}\theta$};


        \def\lowtheta{0.25}
        \def\hightheta{0.75}

        \def\gapshortstart{0.025}
        \def\gapshortend{0.025}
        \def\gaplongstart{0.048}
        \def\gaplongend{0.048}

        \pgfmathsetmacro{\xoneL}{0.4+0.2*\lowtheta}
        \pgfmathsetmacro{\xhighL}{0.8-0.2*\lowtheta}
        \pgfmathsetmacro{\xlowL}{0.4-0.2*\lowtheta}

        \pgfmathsetmacro{\xoneH}{0.4+0.2*\hightheta}
        \pgfmathsetmacro{\xhighH}{0.8-0.2*\hightheta}
        \pgfmathsetmacro{\xlowH}{0.4-0.2*\hightheta}

        \fill[black] (\lowtheta,\xoneL) circle[radius=1.5pt];
        \fill[black] (\hightheta,\xoneH) circle[radius=1.5pt];

        \draw (\lowtheta,0) -- (\lowtheta,-0.015);
        \draw (\hightheta,0) -- (\hightheta,-0.015);

        \draw[-Stealth, blue]
            (\lowtheta,\xoneL+\gaplongstart) -- (\lowtheta,\xhighL-\gaplongend);

        \draw[-Stealth, purple]
            (\lowtheta,\xoneL-\gapshortstart) -- (\lowtheta,\xlowL+\gapshortend);

        \draw[-Stealth, blue]
            (\hightheta,\xoneH+\gapshortstart) -- (\hightheta,\xhighH-\gapshortend);

        \draw[-Stealth, purple]
            (\hightheta,\xoneH-\gaplongstart) -- (\hightheta,\xlowH+\gaplongend);

        \node[note,anchor=north] at (0.25,-0.03)
            {low-quality project};

        \node[note,anchor=north] at (0.75,-0.03)
            {high-quality project};


        \node[note, blue, above] at (\lowtheta,0.78)
            {large\\increase};

        \node[note, purple, below] at (\lowtheta,0.32)
            {small\\punishment};

        \node[note, blue, above] at (\hightheta,0.68)
            {small\\increase};

        \node[note, purple, below] at (\hightheta,0.250)
            {large\\punishment};

    \end{tikzpicture}
    \caption{The optimal mechanism for the example}
    \label{fig:example}
\end{figure}

Figure~\ref{fig:example} illustrates the optimal mechanism in this example. The first-period allocation $x_1^*$ increases with project quality $\theta$, while both second-period allocations $x_{2h}^*$ and $x_{2l}^*$ decrease. Consider a low-quality project with $\theta=\frac{1}{4}$. It receives an initial allocation of $0.45$. Following a high signal, its allocation rises substantially to $0.75$; following a low signal, it falls only slightly to $0.35$. In contrast, a high-quality project with $\theta=\frac{3}{4}$ receives a larger initial allocation of $0.55$. Following a high signal, its allocation rises only slightly to $0.65$; following a low signal, it falls substantially to $0.25$. 

The principal thus has higher expectations for higher-quality projects and treats them more strictly: a high signal leads to a small increase in resources, whereas a low signal leads to a large reduction. Conversely, lower-quality projects receive fewer resources initially, but the principal has lower expectations for them and is therefore more lenient: they face a small reduction after a low signal but receive a substantial increase after a surprisingly high signal.

To see how the optimal mechanism uses the midterm review information, we compare it with the static optimal mechanism in Figure~\ref{fig:comparison}.
The principal wants to match allocated resources to project quality, so her most preferred action is the 45-degree line.
Compared with the static optimal benchmark, the optimal mechanism brings the first-period allocation closer to the most preferred action: $x^*_1$ is obtained by rotating $a_1$ toward the 45-degree line around the intersection point. 
Incentive compatibility then implies that the second-period allocations cannot both be increasing in project quality. The optimal mechanism makes $x^*_{2h}$ and $x^*_{2l}$ less aligned with the ideal action: they are obtained by rotating $a_{2h}$ and $a_{2l}$ away from the 45-degree line around the respective intersection points, by the same angle that $x^*_1$ is rotated toward it.

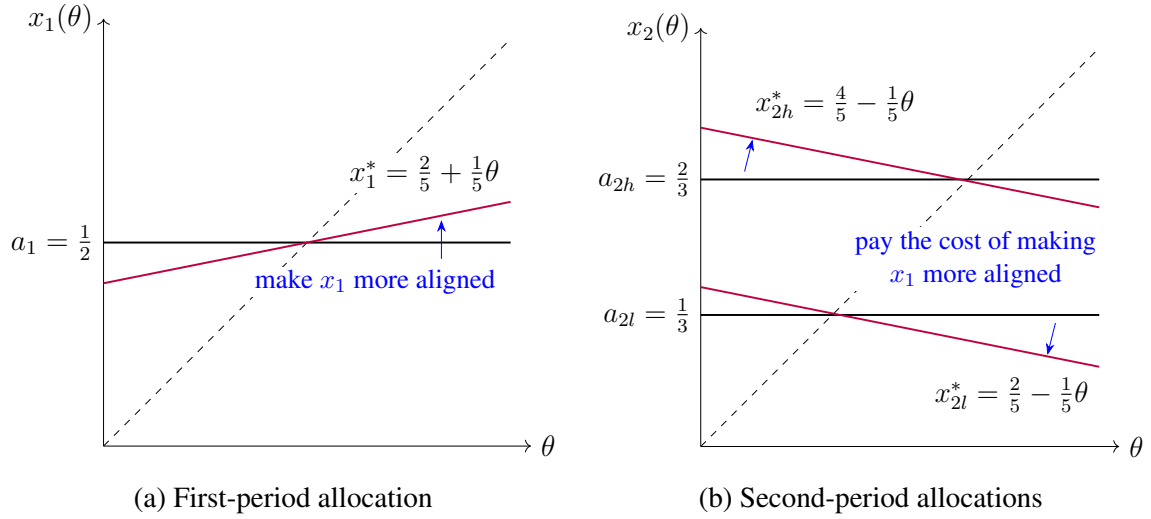
\begin{figure}[htbp]
    \centering
    \captionsetup[subfigure]{justification=centering}

    \begin{subfigure}[t]{0.49\textwidth}
        \centering
        \resizebox{\linewidth}{!}{%
        \begin{tikzpicture}[x=5.8cm,y=5.8cm]

            \draw[->] (0,0) -- (1.05,0)
                node[right] {$\theta$};
            \draw[->] (0,0) -- (0,1.05)
                node[left] {$x_1(\theta)$};

            \draw[dashed] (0,0) -- (1,1);

            \draw[thick] (0,0.5) -- (1,0.5);
            \node[left] at (0,0.5) {$  a_{1} = \frac{1}{2}$};

            \draw[thick,purple] (0,0.4) -- (1,0.6);

            \node[
                anchor=south east,
                fill=white,
                inner sep=1pt
            ] at (0.98,0.62)
                {$x^*_{1} = \frac{2}{5} + \frac{1}{5}\theta$};

            \draw[-Stealth,blue]
                (0.83,0.46) -- (0.83,0.555);

            \node[
                blue,
                font=\small,
                anchor=north east,
                fill=white,
                inner sep=1pt
            ] at (0.97,0.44)
                {make $x_1$ more aligned};

        \end{tikzpicture}%
        }

        \caption{First-period allocation}
        \label{fig:levering}
    \end{subfigure}%
    \hfill
    \begin{subfigure}[t]{0.49\textwidth}
        \centering
        \resizebox{\linewidth}{!}{%
        \begin{tikzpicture}[x=5.8cm,y=5.8cm]

            \draw[->] (0,0) -- (1.05,0)
                node[right] {$\theta$};
            \draw[->] (0,0) -- (0,1.05)
                node[left] {$x_2(\theta)$};

            \draw[dashed] (0,0) -- (1,1);

            \draw[thick] (0,0.67) -- (1,0.67);
            \node[left] at (0,0.67) {$a_{2h} = \frac{2}{3} $};

            \draw[thick] (0,0.33) -- (1,0.33);
            \node[left] at (0,0.33) {$a_{2l} = \frac{1}{3} $};

            \draw[thick,purple] (0,0.8) -- (1,0.6);

            \node[
                anchor=south west,
                fill=white,
                inner sep=1pt
            ] at (0.13,0.81)
                {$x^*_{2h} = \frac{4}{5} - \frac{1}{5}\theta$};

            \draw[thick,purple] (0,0.4) -- (1,0.2);

            \node[
                anchor=north east,
                fill=white,
                inner sep=1pt
            ] at (0.98,0.18)
                {$x^*_{2l} = \frac{2}{5} - \frac{1}{5}\theta$};

            \draw[-Stealth,blue]
                (0.11,0.69) -- (0.13,0.77);

            \draw[-Stealth,blue]
                (0.89,0.31) -- (0.87,0.23);


            \node[
                blue,
                align=center,
                font=\small,
                fill=white,
                inner sep=1.5pt
            ] at (0.69,0.47)
                {pay the cost of making\\
                 $x_1$ more aligned};

        \end{tikzpicture}%
        }

        \caption{Second-period allocations}
        \label{fig:payingcost}
    \end{subfigure}

    \caption{Comparison of the optimal mechanism and the static optimal mechanism.}
    \label{fig:comparison}
\end{figure}

The optimal mechanism can therefore be interpreted as frontloading the review's information value to the first period. Relative to the static optimal benchmark, it brings the first-period allocation closer to the principal's ideal action, at the cost of moving the second-period allocations farther away in order to preserve incentive compatibility. Why is this intertemporal trade-off worthwhile? The midterm review reduces uncertainty: the posterior variance following either signal is smaller than the prior variance. Under quadratic loss, the expected gain or loss from such a rotation is proportional to the corresponding variance. Because the posterior variances are smaller, the gain from improving the first-period allocation outweighs the cost of distorting the second-period allocations.

In fact, this logic extends beyond the example. For any distribution and signal function, sufficiently small rotations of this kind locally improve the principal's expected payoff relative to the static optimal benchmark, even though the globally optimal mechanism need not take this geometric form. Nevertheless, we show that the more general property remains: the optimal mechanism brings forward the review's information value, and does so to the extent that the principal's expected payoff in the first period strictly exceeds that in the second period. We discuss this result in more detail in Section~\ref{sec:mainresults}.

Another feature of the optimal mechanism is that 
\[
x^*_{2h}(\theta) \geq x^*_1 (\theta) \geq x^*_{2l}(\theta) \quad\text{for all }\theta \in \Theta.
\]
That is, relative to the first-period allocation, the optimal mechanism always rewards the agent after an $H$ signal and punishes the agent after an $L$ signal, though the magnitude of the reward or punishment depends on project quality.

 In fact, the entire family of Beta distributions admits closed-form solutions for the optimal mechanism, with the uniform distribution as the special case $\alpha=\beta=1$. As in the uniform case, all three allocation rules are linear in project quality, with slopes and intercepts determined by the parameters of the Beta distribution. Proposition~\ref{prop:beta-prior} formalizes this result.\footnote{Proposition~\ref{prop:beta-prior} allows the density to vanish or diverge at the endpoints, relaxing the maintained assumption that $f$ be continuous and strictly positive on $\Theta$. Its proof is by direct verification and does not rely on the general existence result.}
 Specifically, for a $\mathrm{Beta}(\alpha,\beta)$ distribution, $x_1^*$ has slope $\frac{1}{2(\alpha+\beta)+1}$, while $x_{2h}^*$ and $x_{2l}^*$ both have slope $-\frac{1}{2(\alpha+\beta)+1}$. The optimal mechanism can again be obtained by rotating the static optimal mechanism around its intersection points with the 45-degree line.

\begin{proposition}[Optimal mechanism under Beta priors]\label{prop:beta-prior}
Let the signal probability be $s(\theta)=\theta$, and suppose that the type is distributed according to $\theta\sim\mathrm{Beta}(\alpha,\beta)$ with $\alpha,\beta>0$, i.e.
$f(\theta)=\theta^{\alpha-1}(1-\theta)^{\beta-1}/B(\alpha,\beta)$. Then the unique optimal mechanism is
\begin{equation}\label{eq:betamech}
x^*_1(\theta)=\frac{2\alpha+\theta}{2(\alpha+\beta)+1},\quad
x^*_{2h}(\theta)=\frac{2\alpha+2-\theta}{2(\alpha+\beta)+1},\quad
x^*_{2l}(\theta)=\frac{2\alpha-\theta}{2(\alpha+\beta)+1}.
\end{equation}
\end{proposition}

Note that when $\alpha<\frac{1}{2}$, $x_{2l}^*$ can be negative for sufficiently high-quality projects. Holding $\beta$ fixed, a smaller $\alpha$ places more probability mass on low-quality projects. In this case, the principal may claw back funding from projects that report high quality but receive a low signal. This severe punishment helps deter the large mass of low-quality projects from exaggerating their quality.

Hence, the properties of the optimal mechanism identified in the uniform case extend to all Beta distributions. Specifically, the first-period allocation is increasing in quality, both second-period allocations are decreasing, and the first-period allocation is sandwiched between them, so an $H$ signal always increases resources while an $L$ signal always reduces them.
Moreover, the optimal mechanism frontloads the midterm review's information value to the first period. In Section~\ref{sec:mainresults}, 
we show that many of these properties extend beyond the Beta family to general type distributions and signal functions.

\section{Properties of the Optimal Mechanism}
\label{sec:mainresults}

We now return to general type distributions and signal functions. Note that the principal's problem can be expressed as an optimal-control problem over the type space. Take the agent's 
interim expected payoff $u$, the first-period allocation $x_1$, and the 
signal-contingent spread $\Delta$ as the choice variables. The two second-period allocations can be recovered from these choice variables. By Lemma~\ref{lem:IC}, the principal's problem is equivalent to

\begin{align*}
\min_{u,x_1,\Delta}\quad
&\int_{\Theta}
\Big[
    \bigl(x_1(\theta)-\theta\bigr)^2
    +\bigl(u(\theta)-x_1(\theta)-\theta\bigr)^2
    +s(\theta)\bigl(1-s(\theta)\bigr)\Delta(\theta)^2
\Big]\,dF(\theta)\\
\text{s.t.}\quad
&u'(\theta)=s'(\theta)\Delta(\theta),\\
&\Delta \text{ is increasing on } \Theta,\\
&u(\theta)\geq 0 \quad  \forall \; \theta\in\Theta.
\end{align*}

Here $u$ is the state variable, while $x_1$ and $\Delta$ are controls.
The envelope condition governs the 
evolution of the state across types, and incentive compatibility additionally 
requires the spread to be increasing. 
We first establish existence and uniqueness.

\begin{lemma} \label{lem:existence}
    (Existence and uniqueness)
There exists a unique optimal mechanism.
\end{lemma}

The uniqueness result is standard: the feasible set of mechanisms is convex, and the principal's expected payoff is strictly concave in the allocation rules. Establishing existence is more involved. A standard compactness argument can establish an optimum in a larger class of measurable allocation rules for which IC and IR hold almost everywhere. However, this does not by itself ensure that the optimum can be represented by finite, continuous allocation rules satisfying both constraints at every type. We close this gap in two steps. First, we establish uniform bounds that allow us to obtain finite allocation rules satisfying IC and IR everywhere. Second, we use a local improvement argument to show that the optimum admits continuous allocation rules. Together, these steps establish existence in the original admissible class. The full proof is provided in the Appendix.

Beyond special cases such as Proposition \ref{prop:beta-prior}, a closed-form solution is
generally unavailable even when the monotonicity and participation
constraints are omitted. The relaxed problem still requires solving
a second-order differential equation whose coefficients depend on
the type distribution and signal function. Specifically, after minimizing pointwise over $x_1$ and using the envelope condition, the Euler--Lagrange condition for the relaxed problem is
\[
\frac{d}{d\theta}
\left[
    \frac{2s(\theta)(1-s(\theta))f(\theta)}
         {(s'(\theta))^2}
    u'(\theta)
\right]
=
f(\theta)\bigl(u(\theta)-2\theta\bigr),
\]
together with the corresponding endpoint conditions.

 Nevertheless, we can still characterize the qualitative structure of the optimal mechanism through several general properties. We first show that the mechanism frontloads the review's information value, and then characterize how allocations vary with project quality and respond to the review outcome.

\subsection{Frontloading information value}
\label{sec:frontload}

In the example above, we used the term ``frontloading information value'' to describe the phenomenon that, relative to the static optimal mechanism, the optimal mechanism brings the first-period allocation closer to the principal's ideal action while moving the second-period allocations farther away. This process effectively shifts part of the review's value from the second period to the first. To formalize the extent of this intertemporal trade-off, we say that a mechanism frontloads information value if the principal's expected payoff in the first period exceeds that in the second.

\begin{definition}
    \label{def:levering}
    A mechanism frontloads information value if the principal's expected payoff in the first period is higher than that in the second period.
\end{definition}

\begin{theorem}
\label{thm:frontloading-information}
The optimal mechanism frontloads information value.
\end{theorem}

We provide a sketch of the proof here; the full proof is in the Appendix. We first show that, for every project quality, the optimal mechanism allocates the same expected amount of resources in the two periods.
Define $ m^*(\theta)
    \equiv
    s(\theta)x^*_{2h}(\theta)
    +
    \bigl(1-s(\theta)\bigr)x^*_{2l}(\theta)$ as the expected second-period allocation for a type-$\theta$ project under the optimal mechanism. We show that $x_1^*(\theta) = m^*(\theta)$ for all $\theta \in \Theta$. 
    To see this, consider the perturbation $\{x^*_1 - \epsilon \phi, x^*_{2h} + \epsilon \phi, x^*_{2l} + \epsilon \phi\}$, where $\epsilon \in \mathbb{R}$ and $\phi: \Theta \to \mathbb{R}$ is any continuous function. This perturbation leaves the agent's expected payoff from every report unchanged and therefore preserves incentive compatibility and individual rationality. The principal's expected payoff is strictly concave in $\epsilon$, so optimality requires the first-order condition at $\epsilon = 0$:
\[
    \int_{\Theta}
    \left(
        x_1^*(\theta)
        -
        m^*(\theta)
    \right)
    \phi(\theta)\,dF(\theta)
    = 0
\]
for every continuous function $\phi$. Therefore, we have 
\[
    x_1^*(\theta) = m^*(\theta) \quad\text{for all }\theta \in \Theta, \tag{balancing condition}
\]
which we refer to as the \emph{balancing condition}. It states that, under quadratic loss, the optimal division of resources across periods requires allocating the same expected amount in each period.\footnote{A related equalization of current and expected future distortions appears in \cite{chen2022dynamic}. Here, because project quality
is fixed, the balancing condition equates expected allocations
across the two periods for each type.}

Let  $L_1^*$ and $L_2^*$ denote the principal's expected loss in the first and second periods under the optimal mechanism, respectively: 
\begin{align*}
    L_1^* & \equiv \int_{\Theta} \bigl(x_1^*(\theta)-\theta\bigr)^2\,dF(\theta), \\
    L_2^* & \equiv \int_{\Theta} \left[ s(\theta)\bigl(x_{2h}^*(\theta)-\theta\bigr)^2 + (1-s(\theta))\bigl(x_{2l}^*(\theta)-\theta\bigr)^2\right]\,dF(\theta).
\end{align*}
For each type $\theta$, the principal's second-period expected loss under the optimal mechanism can be written as
\begin{align*}
    \bigl(m^*(\theta)-\theta\bigr)^2 + s(\theta)\bigl(1-s(\theta)\bigr) \Delta^*(\theta)^2 
  = \bigl(x_1^*(\theta)-\theta\bigr)^2 + s(\theta)\bigl(1-s(\theta)\bigr) \Delta^*(\theta)^2. 
\end{align*}
It follows that
\begin{equation}\label{eq:losscomparison}
        L_2^* = L_1^*
    +
    \mathbb{E} \left[  s(\theta)\bigl(1-s(\theta)\bigr) \Delta^*(\theta)^2 \right] \geq L_1^*.
\end{equation}
Thus, under the optimal mechanism, the principal's expected loss is higher in the second period than in the first. The balancing condition equalizes the mean allocation across periods, but the second-period allocation still fluctuates around this mean with the review signal, generating additional variance. The term $\mathbb{E} \left[  s(\theta)\bigl(1-s(\theta)\bigr) \Delta^*(\theta)^2 \right]$  captures precisely this cost of signal-contingent variation.

The Appendix also establishes that $\Delta^* \not\equiv 0$, so the inequality in \eqref{eq:losscomparison} is strict. Thus, we obtain a reversal result: the principal's expected payoff is strictly higher in the first period than in the second, even though she is better informed in the second.

To understand the economic logic behind frontloading, recall that the agent cares only about total resources, whereas the principal cares about how those resources are divided across periods, since she aims to match resources to project quality in each period. The principal therefore prefers to shift allocations across periods to balance marginal losses. The midterm review makes this possible by providing an instrument to elicit information from the agent. 
Equalizing expected marginal losses does not, by itself, establish the required comparison of expected loss levels; that step also depends on the loss function. Quadratic loss makes the connection particularly transparent through the variance decomposition above. In fact, the frontloading results extend beyond quadratic loss, as stated in Remark \ref{rmk:generallossfunctions} and generalized further in the Appendix.\footnote{For a continuously differentiable, strictly convex loss function $l$, define $\Gamma = l \circ (l')^{-1}$, which gives the loss level associated with a given marginal loss. When $\Gamma$ is convex, Jensen's inequality
implies that marginal-loss equalization yields $L_1^* \leq L_2^*$.
Proposition~\ref{prop:generallossfunctions} in the
Appendix formalizes this sufficient condition for frontloading and also
considers losses that are concave functions of absolute error.}

\begin{remark}
\label{rmk:generallossfunctions}

For loss functions of the form $l(x,\theta)=|x-\theta|^p$, the optimal mechanism frontloads information value when $p>1$. When $0<p\leq 1$, there exists an optimal mechanism that frontloads information value.

\end{remark}

In the example above, the optimal mechanism frontloads information value by rotating the static optimal mechanism. 
For general distributions and signal functions, however, the optimal mechanism need not take this form. In the next subsection, we establish several properties of the optimal mechanism that shed light on its general structure.

\subsection{Monotonicity and the sandwich property}

In this section, we first establish the monotonicity properties of the optimal allocation rules in Theorem~\ref{thm:general-monotonicity}, and then characterize the relationship between the first-period and second-period allocations in Theorem~\ref{thm:sandwich}, as well as boundary conditions in Proposition \ref{prop:boundaryconditions}. Finally, we combine these results to illustrate the structure of the optimal mechanism through examples.

\begin{theorem}
\label{thm:general-monotonicity}
The optimal mechanism has the following properties:
\begin{enumerate}[(i)]
    \item $x_1^*$ is increasing on $\Theta$.
    \item $x_{2l}^*$ is decreasing on $\Theta$.
    \item $x_{2h}^*$ cannot be increasing on all of $\Theta$.
At any type $\theta$ where $\Delta^*$ is differentiable,
     \[
    {x_{2h}^*}'(\theta)
    =
    -\frac{1}{2}\,
    s'(\theta)\,\Delta^*(\theta)
    +
    \bigl(1-s(\theta)\bigr)\,{\Delta^*}'(\theta).
    \]
    If $s(\bar{\theta})<1$, or  the allocation rules are Lipschitz continuous, then there exists $\theta_0<\bar{\theta}$ such that $x_{2h}^*$ is strictly decreasing on $[\theta_0,\bar{\theta}]$.
\end{enumerate}
\end{theorem}

Theorem~\ref{thm:general-monotonicity} shows that, as in the example, the optimal mechanism frontloads  information value so that the first-period allocation is more aligned with the principal's objective: higher-quality projects receive more resources. Incentive compatibility then implies an intertemporal trade-off between the slopes of the first-period and second-period allocations. To see this, consider differentiable allocation rules, incentive compatibility requires
\[
s(\theta){x_{2h}^*}'(\theta) + \bigl(1-s(\theta)\bigr){x_{2l}^*}'(\theta) = -{x_1^*}'(\theta).
\]
The slopes reflect how responsive the allocations are to the agent's type. Since the optimal mechanism makes the first-period allocation responsive to type, it must offset this by making the second-period allocations negatively responsive in the probability-weighted sense.\footnote{The offsetting of slopes across periods is reminiscent of \cite{chen2022dynamic}, where incentive compatibility implies that any increase in the current action must be matched by lower future actions, so that the optimal contract can respond to the agent's report in a contrarian direction.} Specifically, the optimal mechanism makes $x_{2l}^*$ decreasing in type, while the monotonicity of $x_{2h}^*$ is ambiguous because it is shaped by two opposing forces.

To understand these monotonicity properties and the forces behind them, suppose for now that $\Delta^*(\theta)\geq 0$ on $\Theta$. This means that the optimal mechanism always allocates more resources after an $H$ signal than after an $L$ signal: $x^*_{2h}(\theta)\geq x^*_{2l}(\theta)$ for all $\theta\in\Theta$. While this may seem natural, it is not obvious that the optimal mechanism must have this property. We take it as given for now and establish it formally in Proposition~\ref{prop:nonnegative-spread}, with a sketch of the proof later in this section.

In Section \ref{sec:frontload}, we established the balancing condition, which implies that the first-period allocation equals exactly half of the agent's interim expected payoff: $x_1^*(\theta) = \frac{1}{2}u^*(\theta)$. The envelope condition in Lemma~\ref{lem:IC} gives $u'(\theta) = s'(\theta)\Delta(\theta)$, so
\begin{equation}
    \label{eq:x1}
    {x_{1}^*}'(\theta) = \frac{1}{2}{u^*}'(\theta) = \frac{1}{2}s'(\theta)\Delta^*(\theta) \geq 0.
\end{equation}
Thus, $x_1^*$ is increasing because the agent's interim expected payoff is increasing in type. Intuitively, higher-quality projects are more likely to generate an $H$ signal, so incentive compatibility implies that they earn larger interim expected payoffs, which, through the balancing condition, translates into larger first-period allocations.

For the second-period allocations, two forces are at play. They work in the same direction for $x_{2l}^*$, making it decreasing, but in opposite directions for $x_{2h}^*$, leaving its monotonicity ambiguous. Specifically, the balancing condition implies
\begin{equation*}
    x_{2l}^*(\theta) = x_1^*(\theta) - s(\theta)\Delta^*(\theta).
\end{equation*}
Differentiating this identity and using \eqref{eq:x1} yields the derivative of $x_{2l}^*$:
\[
{x_{2l}^*}'(\theta)
=
\underbrace{
    -\frac{1}{2}s'(\theta)\Delta^*(\theta)
}_{\text{surprise effect}\leq 0}
\;
\underbrace{
    \vphantom{-\frac{1}{2}s'(\theta)\Delta^*(\theta)}
    -s(\theta){\Delta^*}'(\theta)
}_{\text{incentive effect}\leq 0}
\leq 0.
\]
The derivative consists of two terms. The first term captures the \emph{surprise effect}: holding the spread $\Delta^*(\theta)$ fixed, an $L$ signal is less likely---and therefore more surprising---for higher-quality projects, so they should be punished more after an $L$ signal. This force pushes $x_{2l}^*$ downward as project quality increases. The second term captures the \emph{incentive effect}: holding the signal probability $s(\theta)$ fixed, incentive compatibility requires the spread $\Delta^*(\theta)$ to increase with type. A wider spread  pushes $x_{2h}^*$ upwards and $x_{2l}^*$ downwards as $\theta$ increases. Thus, these two forces work in the same direction for $x_{2l}^*$, making the allocation rule decreasing.

Similarly, the derivative of $x_{2h}^*$ is given by:
\[
{x_{2h}^*}'(\theta)
=
\underbrace{
    -\frac{1}{2}s'(\theta)\Delta^*(\theta)
}_{\text{surprise effect}\leq 0}
\; +
\underbrace{
    \vphantom{-\frac{1}{2}s'(\theta)\Delta^*(\theta)}
    \bigl(1-s(\theta)\bigr){\Delta^*}'(\theta)
}_{\text{incentive effect}\geq 0}.
\]

The first term captures the same surprise effect as for $x_{2l}^*$: holding the spread $\Delta^*(\theta)$ fixed, an $H$ signal is more likely---and therefore less surprising---for higher-quality projects, so they should receive a smaller reward after an $H$ signal. This force pushes $x_{2h}^*$ downward as project quality increases.  
The second term again captures the incentive effect. Incentive compatibility requires the gap between $x_{2h}^*$ and $x_{2l}^*$ to widen with project quality, which pushes $x_{2h}^*$ upward. Thus, these two forces work in opposite directions for $x_{2h}^*$, leaving its monotonicity ambiguous in general.

Nonetheless, we show that $x_{2h}^*$ cannot be increasing near the highest type. To see the intuition, note that the term $1-s(\theta)$ in the incentive effect is decreasing in $\theta$, so for the incentive effect to dominate the surprise effect and make $x_{2h}^*$ increasing, the spread $\Delta^*(\theta)$ would have to grow sufficiently rapidly. Such rapid growth, however, is not optimal for high enough types, because there are almost no higher types left whose allocations and incentives can benefit from further widening of the spread.\footnote{Specifically, we show that the principal can strictly improve her payoff by replacing a rapidly growing spread with a constant spread for sufficiently high types.} 
This logic is reminiscent of no distortion at the top in standard mechanism design, as there are no higher types left to screen. When the highest type is not guaranteed to receive a high signal, i.e., $s(\bar{\theta})<1$, or when we impose an additional regularity condition that the optimal allocation rules are Lipschitz continuous, then ${\Delta^*}'(\theta)$ is bounded, and we can show that there exists a terminal interval of types on which $x_{2h}^*$ is strictly decreasing. In this sense, the analogue of no distortion at the top in our setting concerns the growth of the spread rather than the first-best allocation for the highest type.

The derivative formula for $x_{2h}^*$ in Theorem~\ref{thm:general-monotonicity} depends on the endogenous object $\Delta^*(\theta)$, which is determined by the optimal mechanism.  In general, it is difficult to derive sufficient conditions for the monotonicity of $x_{2h}^*$ directly from primitives without imposing specific functional forms on the type distribution and signal function. 
This is because the sign of ${x_{2h}^*}'(\theta)$ generally depends on the global solution for the optimal spread, not only on local quantities such as $f(\theta)$, $s(\theta)$, and their derivatives.\footnote{Determining the optimal spread requires solving a differential equation involving the type distribution and the signal function.}  Nonetheless, the derivative formula is useful for understanding the two counteracting forces behind the monotonicity of $x_{2h}^*$.

The uniform distribution---and, more generally, the Beta family---is a special case in which the optimal spread $\Delta^*(\theta)$ is constant. Hence ${\Delta^*}'(\theta)=0$, and the incentive effect disappears. Only the surprise effect remains, implying that both $x_{2l}^*$ and $x_{2h}^*$ are decreasing in $\theta$.

Next, we sketch the proof that the optimal spread $\Delta^*(\theta)$ is nonnegative (Proposition~\ref{prop:nonnegative-spread}). This result will also imply the next property of the optimal mechanism, the \emph{sandwich property}: $x_{2l}^* \leq x_1^* \leq x_{2h}^*$ on $\Theta$.  The full proof is provided in the Appendix.

\begin{proposition}[Nonnegative optimal spread] $\Delta^*(\theta)\geq0 \; \text{ for all }\theta\in\Theta.$
\label{prop:nonnegative-spread}
\end{proposition}

In this model, the signal probability $s(\theta)$ serves as the effective screening type. In particular, two types with the same signal probability are observationally equivalent from the perspective of incentive provision and therefore must be treated identically. We thus use $y \equiv s(\theta)$ as the screening type and define the inverse function $\psi(y) \equiv s^{-1}(y)$. The agent's interim expected payoff can then be expressed as a function of the screening type $y$:
\[
u(\theta) = u(\psi(y)) \equiv U(y).
\]
The incentive compatibility condition in Lemma \ref{lem:IC} gives $u'(\theta) = s'(\theta)\Delta(\theta)$. Combining this with the chain rule yields
\begin{align*}
    U'(y) =  u'(\psi(y))\psi'(y)  = s'(\psi(y))\Delta(\psi(y))\frac{1}{s'(\psi(y))} = \Delta(\psi(y)),
\end{align*}
so that
\begin{equation}
    \label{eq:Uprime}
    U'(y) = \Delta(\psi(y)).
\end{equation}
Incentive compatibility requires the spread $\Delta$ to be increasing. Since $\psi$ is  strictly increasing, it follows that $U$ is convex in $y$.

Using the balancing condition, the principal's expected loss can be rewritten as
\begin{equation}
\label{eq:reduced-objective-intuition}
    \mathcal{J}(U)
    =  \int_{\underline{y}}^{\bar{y}}
    \left[
        \frac{1}{2}
        \bigl(
            U(y)-2\psi(y)
        \bigr)^2
        +
        y(1-y)\bigl(U'(y)\bigr)^2
    \right]
    g(y)\,dy,
\end{equation}
where $\underline{y}=s(\underline{\theta})$, $\bar{y}=s(\bar{\theta})$, and $g(y)$ denotes the density of the screening type $y$. In the first term inside the brackets, $2\psi(y)=2\theta$ is the principal's ideal total allocation across the two periods. Hence, $ \frac{1}{2}\big(U(y)-2\psi(y)\big)^2$ measures the gap between the ideal and the actual expected total allocation.

Now suppose that the optimal mechanism has $\Delta^*(\theta)<0$ at some type. By continuity, there exists an interval on which $\Delta^*(\theta)<0$. Equation~\eqref{eq:Uprime} then implies that $U^*(y)$ is decreasing on that interval. We can then construct an alternative mechanism by replacing this decreasing segment of $U^*(y)$ with a flat segment. This modification strictly increases the principal's expected payoff. Figure~\ref{fig:nonnegative-spread} illustrates the argument.

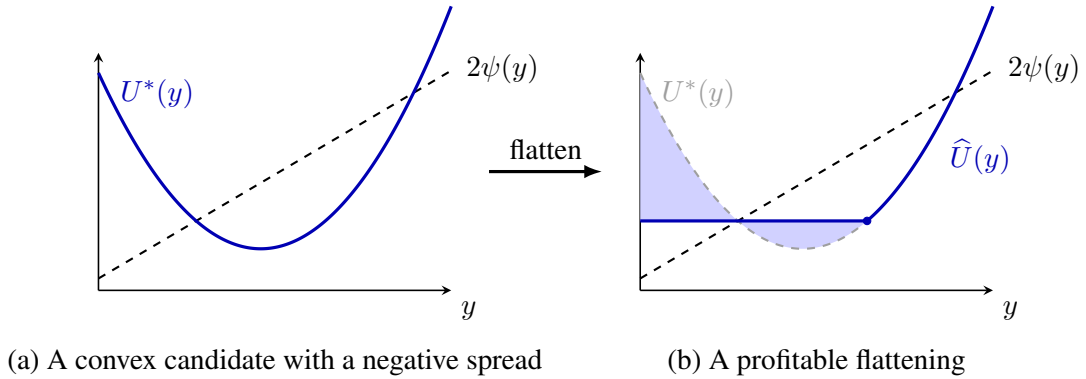
\begin{figure}[htbp]

\centering

\begin{tikzpicture}

\begin{groupplot}[
    group style={
        group size=2 by 1,
        horizontal sep=2.50cm
    },
    width=0.41\textwidth,
    height=0.31\textwidth,
    xmin=0,
    xmax=5,
    ymin=0.55,
    ymax=2.55,
    axis lines=left,
    axis line style={semithick},
    tick style={draw=none},
    xtick=\empty,
    ytick=\empty,
    clip=false,
    xlabel={$y$},
    xlabel style={
        at={(axis description cs:1.05,0)},
        anchor=north
    },
    title style={
        at={(axis description cs:0.5,-0.47)},
        anchor=south,
        font=\small
    },
    every axis label/.style={font=\small},
]


\nextgroupplot[
    title={(a) A convex candidate with a negative spread}
]

\addplot[
    black,
    dashed,
    thick,
    domain=0:5,
    samples=2
]
{0.65+0.35*x}
node[
    pos=1.01,
    right,
    font=\small
]
{$2\psi(y)$};

\addplot[
    blue!70!black,
    very thick,
    domain=0:5,
    samples=180
]
{0.28*(x-2.3)^2+0.9}
node[
    pos=0.035,
    right,
    font=\small
]
{$U^*(y)$};


\nextgroupplot[
    title={(b) A profitable flattening}
]

\addplot[
    black,
    dashed,
    thick,
    domain=0:5,
    samples=2
]
{0.65+0.35*x}
node[
    pos=1.01,
    right,
    font=\small
]
{$2\psi(y)$};

\addplot[
    name path=Uoriginal,
    gray!75,
    dashed,
    thick,
    domain=0:5,
    samples=180
]
{0.28*(x-2.3)^2+0.9}
node[
    pos=0.035,
    right,
    font=\small,
    text=gray!75
]
{$U^*(y)$};

\addplot[
    name path=Uflat,
    blue!70!black,
    very thick,
    domain=0:3.2154,
    samples=2
]
{1.13462};

\addplot[
    blue!70!black,
    very thick,
    domain=3.2154:5,
    samples=100
]
{0.28*(x-2.3)^2+0.9}
node[
    pos=0.5,
    below right,
    font=\small
]
{$\widehat U(y)$};

\addplot[
    blue!18
]
fill between[
    of=Uoriginal and Uflat,
    soft clip={domain=0:3.2154}
];

\addplot[
    only marks,
    mark=*,
    mark size=1.4pt,
    blue!70!black
]
coordinates {(3.2154,1.13462)};

\end{groupplot}

\draw[
    -{Latex[length=2.4mm]},
    very thick
]
($(group c1r1.east)+(0.50,0)$)
--
($(group c2r1.west)+(-0.50,0)$)
node[
    midway,
    above,
    font=\small
]
{flatten};

\end{tikzpicture}

\caption{
Illustration of the nonnegative-spread argument.
}
\label{fig:nonnegative-spread}

\end{figure}

In Figure~\ref{fig:nonnegative-spread}, the dashed line represents the principal's ideal total allocation $2\psi(y)$, which is increasing in $y$. Panel~(a) shows a convex candidate $U^*(y)$ with a decreasing segment, corresponding to a negative spread $\Delta^*(\theta)<0$. Panel~(b) shows a profitable modification of this candidate, denoted by $\widehat U(y)$. 
Since $U^*(y)$ is convex, the decreasing segment must occur before the minimum. We then replace this decreasing segment, together with any subsequent increasing part needed to preserve convexity, with a flat segment.
The height of the flat segment is chosen so that it lies between the original $U^*(y)$ and the target $2\psi(y)$ over the modified region. The new candidate $\widehat U(y)$ therefore satisfies
\[
    \bigl(
        \widehat U(y)-2\psi(y)
    \bigr)^2
    \leq
    \bigl(
        U^*(y)-2\psi(y)
    \bigr)^2,
\]
so the first term in \eqref{eq:reduced-objective-intuition} is reduced. The Appendix also shows how to choose the height of the flat segment to obtain a profitable modification for all other possible relative positions of $U^*(y)$ and $2\psi(y)$.

Moreover, the second term also decreases on the modified region, it becomes zero since $\widehat U'(y)=0$ there.
The modification thus strictly increases the principal's expected payoff. 
Since $\widehat{U}$ remains convex, it is still incentive compatible.\footnote{Although $\widehat{U}$ may have a kink where the flat segment joins the original function, it can be approximated arbitrarily closely by continuously differentiable convex functions, and the strict improvement survives this approximation.} This contradicts the optimality of $U^*(y)$. Therefore, the optimal mechanism cannot have $\Delta^*(\theta)<0$ at any type.

Under truthful reporting, the balancing condition gives the allocation process a two-period martingale property: the initial allocation equals  the expected continuation allocation. Together with the nonnegative optimal spread established in Proposition \ref{prop:nonnegative-spread}, it implies that $x_1^*$ lies above $x_{2l}^*$ and below $x_{2h}^*$. We refer to this ordering as the \emph{sandwich property}.

\begin{theorem}\label{thm:sandwich}
    (Sandwich Property) In the optimal mechanism,
    \[
        x_{2l}^*(\theta)
    \leq
    x_1^*(\theta)
    \leq
    x_{2h}^*(\theta)
    \qquad
    \text{for every }\theta\in\Theta.
    \]
\end{theorem}

The sandwich property is consistent with the intuition that a high signal is always good news and a low signal is always bad news for the agent's allocation.

So far, we have established the monotonicity and sandwich properties of the optimal mechanism, which together provide a qualitative characterization of its structure. We next consider two examples to illustrate how these properties jointly shape the optimal mechanism.

\begin{figure}[htbp]
  \centering
  \includegraphics[width=\textwidth]{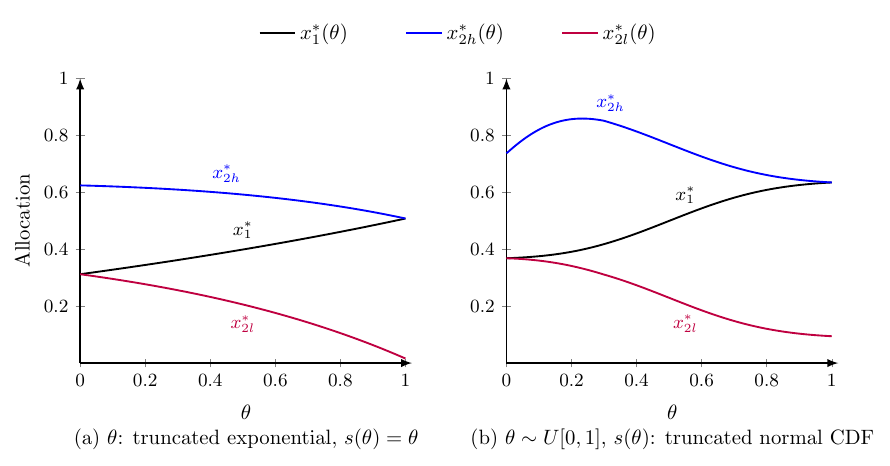}
    \caption{Optimal mechanism for two examples.}
  \label{fig:two-examples}
\end{figure}

Figure~\ref{fig:two-examples} shows the optimal mechanism for two examples.\footnote{For these two examples, the allocation rules are obtained by numerically solving the principal's problem.} The two examples display different monotonicity patterns for $x_{2h}^*$. 
In Panel~(a), the type distribution is the exponential distribution with rate $\lambda=1.5$, truncated to $[0,1]$, i.e.,\ $f(\theta)=\lambda e^{-\lambda\theta}/(1-e^{-\lambda})$ for $\theta\in[0,1]$, and the signal function is $s(\theta)=\theta$. In Panel~(b), the type distribution is uniform on $[0,1]$, and the signal function is S-shaped, given by the CDF of a normal distribution $N(\tfrac12,\tfrac1{16})$ truncated to $[0,1]$. 
Both examples exhibit the general properties established above: $x_1^*$ is increasing, $x_{2l}^*$ is decreasing, and the sandwich property holds. The behavior of $x_{2h}^*$, however, differs across the two examples. In panel~(a), $x_{2h}^*$ is decreasing throughout, whereas in panel~(b), it is nonmonotone, first increasing and then decreasing. In both cases, $x_{2h}^*$ is decreasing near the highest type, consistent with our previous discussion.

Note that the equality $x_1^*(\bar{\theta})=x_{2h}^*(\bar{\theta})$ in these two examples is coincidental: it arises because the highest type receives an $H$ signal with probability one. In general, these two allocations need not coincide. By contrast, the equality $x_1^*(\underline{\theta})=x_{2l}^*(\underline{\theta})$ is a general property of the optimal mechanism. Moreover, if the lowest type has a positive probability of receiving an $H$ signal, i.e., if $s(\underline{\theta})>0$, then
$
x_1^*(\underline{\theta})
=
x_{2l}^*(\underline{\theta})
=
x_{2h}^*(\underline{\theta}).
$
The following proposition summarizes these endpoint properties.

\begin{proposition} \label{prop:boundaryconditions}
The following boundary conditions hold in the optimal mechanism:
\begin{enumerate}
    \item If $s(\underline{\theta})=0$, then $x^*_1(\underline{\theta}) = x_{2l}^*(\underline{\theta})$.
    \item If $s(\underline{\theta})>0$, then $x^*_1(\underline{\theta}) = x_{2l}^*(\underline{\theta}) = x_{2h}^*(\underline{\theta})$.
\end{enumerate}
\end{proposition}

Since $x_1^*(\theta)-x_{2l}^*(\theta) = s(\theta)\Delta^*(\theta),$ the first condition follows immediately.
Now suppose that $s(\underline{\theta})>0$. We show that $\Delta^*(\underline{\theta})=0.$
If instead $\Delta^*(\underline{\theta})>0$, then the positive spread at the bottom generates costly variation across the two review outcomes. In the principal's expected loss in \eqref{eq:reduced-objective-intuition}, this cost is captured by the second term inside the brackets, since $U'=\Delta$. We then construct a local modification that reduces the spread over a small interval near the bottom while preserving incentive compatibility through a small upward adjustment in the expected payoffs of those types. This modification strictly lowers the principal's expected loss, contradicting the optimality of $\Delta^*(\underline{\theta})>0$.
Hence, $\Delta^*(\underline{\theta})=0,$ which implies $x_1^*(\underline{\theta}) = x_{2l}^*(\underline{\theta}) = x_{2h}^*(\underline{\theta}).$

Figure~\ref{fig:one-example} provides an example with $s(\underline{\theta})>0$, in which all three allocation rules coincide at the bottom. In this example, the type distribution is uniform on $[0,1]$, and the signal function is $s(\theta)=0.15+0.75\theta$.\footnote{The kinks in $x_{2h}^*$ and $x_{2l}^*$ around $\theta\approx 0.3$ arise from ironing, where the monotonicity constraint begins to bind.} At the bottom, the optimal mechanism offers a constant-funding option: an agent who reports the lowest type receives the same allocation regardless of the review outcome. Higher reports can secure more funding initially and following a good review, but a bad review can leave them with less funding than this baseline.

 \begin{figure}[htbp]
  \centering
  \includegraphics[width=0.6\textwidth]{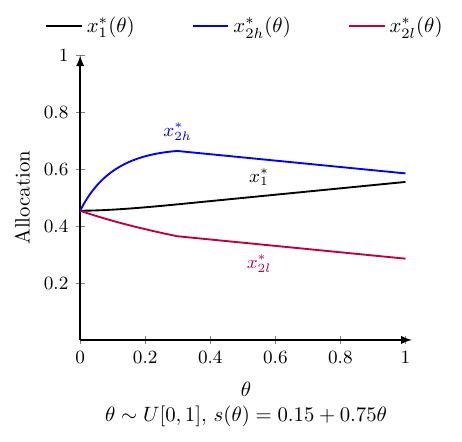}
    \caption{An example with $x_1^*(\underline{\theta}) = x_{2l}^*(\underline{\theta}) = x_{2h}^*(\underline{\theta})$.}
  \label{fig:one-example}
\end{figure}

Combining this boundary condition with the monotonicity results in Theorem~\ref{thm:general-monotonicity} and the sandwich property in Theorem~\ref{thm:sandwich} yields the following corollary.

 \begin{corollary}\label{cor:nonmonotone}
    When $s(\underline{\theta})>0$, $x^*_{2h}$ is nonmonotone.
 \end{corollary}

Read as a funding policy, the optimal mechanism has a simple description. At the outset, the agent reports his project's quality, and the principal allocates more resources to projects that report higher quality. After a bad review, every project is cut, and the cut is deeper for projects that reported higher quality---so much so that the ranking reverses, with higher-quality projects ending up with less funding. After a good review, every project's allocation rises, but the resulting allocation need not preserve the initial ranking. What holds generally is that the allocation after a good review cannot be increasing in quality throughout the entire support and, for Lipschitz continuous allocation rules, it is strictly decreasing among the best projects.

The principal thus holds each project to a standard set by its own report: a higher report earns more resources up front but also exposes the agent to a larger penalty following a disappointing review. This exposure, together with the possibility that higher-quality projects receive lower allocations even after a good review, makes truthful reporting credible and allows the principal to frontload the midterm review's information value.

\section{Conclusion}
\label{sec:conclusion}

In this paper, we study a dynamic delegation problem with a fully biased agent and interim public information, motivated by a resource-allocation setting. We show how the presence of a midterm review can reshape allocations both before and after the review in nontrivial ways. The optimal mechanism frontloads the information value of the review: the principal's expected payoff in the first period is strictly higher than in the second, despite being better informed after the review. Projects that report higher quality receive more resources initially but face less favorable review-contingent continuation terms, receiving less funding after a bad review and potentially less funding even after a good review.

In some applications, the initial allocation may affect the signal probability. For instance, a better-funded project may produce more evidence and therefore be more likely to pass the midterm review. This creates an additional incentive for the agent to exaggerate project quality, since a higher report can now improve both the initial allocation and the probability of receiving a favorable signal. As a result, making the initial allocation more responsive to reported quality requires a larger offset in the probability-weighted continuation slopes. This new force may give the principal a reason to preserve the review's screening power by postponing some resources to the second period.

It is also possible that the principal and agent are informed differently than we have assumed. For instance, the principal may be able to conduct an initial assessment, or the project's type may involve uncertainty so that even the agent cannot fully anticipate its quality at the outset. In these settings, the agent may still possess more information than the principal, and the design of the initial allocation together with review-contingent continuation terms should still help elicit the private information the agent possesses. We conjecture that the forces we have identified remain relevant, though the exact form of the allocation rules may depend on the parties' information structure. We leave the analysis of these extensions to future research.

\begin{appendices}

    \section{Proofs of Results}

    \subsection{Proof of Lemma \ref{lem:IC}}

    \begin{proof}
An agent with true type $\theta$ who reports $\hat{\theta}$ obtains
\begin{align}
    &x_1(\hat{\theta})
    +s(\theta)x_{2h}(\hat{\theta})
    +\bigl(1-s(\theta)\bigr)x_{2l}(\hat{\theta})
    =
    u(\hat{\theta})
    +
    \bigl[
        s(\theta)-s(\hat{\theta})
    \bigr]\Delta(\hat{\theta}).
    \label{eq:lemma1-general-report-payoff}
\end{align}

We first prove necessity. Suppose that the mechanism is incentive
compatible. Fix two types $\theta_2>\theta_1$, and write
$
    s_i\equiv s(\theta_i),
    \Delta_i\equiv\Delta(\theta_i),
    u_i\equiv u(\theta_i),
    i\in\{1,2\}.
$
The incentive constraint for type $\theta_2$ against reporting
$\theta_1$ gives
$
    u_2
    \geq
    u_1+(s_2-s_1)\Delta_1.
$
Similarly, the incentive constraint for type $\theta_1$ against
reporting $\theta_2$ gives
$
    u_1
    \geq
    u_2+(s_1-s_2)\Delta_2.
$
Thus,
\begin{equation}
\label{eq:lemma1-general-IC-bounds}
    (s_2-s_1)\Delta_1
    \leq
    u_2-u_1
    \leq
    (s_2-s_1)\Delta_2.
\end{equation}
Because $s$ is strictly increasing, $s_2-s_1>0$. It follows from
\eqref{eq:lemma1-general-IC-bounds} that
$
    \Delta_1\leq\Delta_2.
$
Hence $\Delta$ is increasing, proving condition {\rm (i)}.

We next derive the Envelope condition. Dividing
\eqref{eq:lemma1-general-IC-bounds} by
$\theta_2-\theta_1>0$ gives
$
    \frac{s_2-s_1}{\theta_2-\theta_1}\Delta_1
    \leq
    \frac{u_2-u_1}{\theta_2-\theta_1} \leq
    \frac{s_2-s_1}{\theta_2-\theta_1}\Delta_2.
$
Fix an interior type $\theta$. Setting $\theta_1=\theta$ and letting
$\theta_2$ converge to $\theta$ from above, continuity of $s'$ and
$\Delta$ implies
$
    u_+'(\theta)
    =
    s'(\theta)\Delta(\theta).
$
Similarly, setting $\theta_2=\theta$ and letting $\theta_1$ converge
to $\theta$ from below gives
$
    u_-'(\theta)
    =
    s'(\theta)\Delta(\theta).
$
Therefore, $u$ is differentiable at every interior type and
$
    u'(\theta)
    =
    s'(\theta)\Delta(\theta) \;
    \text{for every }
    \theta\in\operatorname{int}(\Theta).
$
The corresponding identities hold at
$\underline{\theta}$ and $\bar{\theta}$ using the right and left
derivatives, respectively.

Because $s'\Delta$ is continuous, integrating
$u'(\theta)$ gives
$
    u(\theta)
    =
    u(\underline{\theta})
    +
    \int_{\underline{\theta}}^\theta
    s'(t)\Delta(t)\,dt
$
for every $\theta\in\Theta$. This proves condition {\rm (ii)}.

We next prove sufficiency. Suppose that conditions {\rm (i)} and
{\rm (ii)} hold. Fix an arbitrary true type $\theta$ and an arbitrary
report $\hat{\theta}$. By condition {\rm (ii)},
\begin{align*}
    u(\theta)
    -
    \left\{
        u(\hat{\theta})
        +
        \bigl[
            s(\theta)-s(\hat{\theta})
        \bigr]\Delta(\hat{\theta})
    \right\} &=
    \int_{\hat{\theta}}^\theta
    s'(t)\Delta(t)\,dt
    -
    \Delta(\hat{\theta})
    \int_{\hat{\theta}}^\theta
    s'(t)\,dt
    \\
    &=
    \int_{\hat{\theta}}^\theta
    s'(t)
    \bigl[
        \Delta(t)-\Delta(\hat{\theta})
    \bigr]dt \geq0.
\end{align*}
To see the last inequality, if $\theta>\hat{\theta}$, then
$\Delta(t)\geq\Delta(\hat{\theta})$ over the interval of integration.
If $\theta<\hat{\theta}$, then
$\Delta(t)\leq\Delta(\hat{\theta})$, but the limits of integration are
reversed. Thus, the inequality holds in either case.

Using \eqref{eq:lemma1-general-report-payoff}, the preceding inequality
is equivalent to
$
    u(\theta)
    \geq
    x_1(\hat{\theta})
    +s(\theta)x_{2h}(\hat{\theta})
    +\bigl(1-s(\theta)\bigr)x_{2l}(\hat{\theta}).
$
Hence truthful reporting is optimal for every type, and the mechanism
is incentive compatible.
\end{proof}

\subsection{Proof of Proposition \ref{prop:beta-prior}}

\begin{proof}

We first verify incentive compatibility. Let $n \equiv \alpha+\beta$. Consider an agent with true
type $\theta$ who reports $\hat{\theta}$.  His expected payoff is

\begin{align*}
    x_1^*(\hat{\theta})
    +\theta x_{2h}^*(\hat{\theta})
    +(1-\theta)x_{2l}^*(\hat{\theta})
  &=
    \frac{
        2\alpha+\hat{\theta}
        +\theta(2\alpha+2-\hat{\theta})
        +(1-\theta)(2\alpha-\hat{\theta})
    }{2n+1}\\
&=
    \frac{4\alpha+2\theta}{2n+1}.
\end{align*}
The last expression does not depend on the report $\hat{\theta}$.
Thus, every report is optimal for every type; in particular, truthful
reporting is optimal. Therefore, $x^*$ is incentive compatible.

Under truthful reporting, the agent's interim expected payoff is
$
    u^*(\theta)
    =
    \frac{4\alpha+2\theta}{2n+1}.
$
Since $\alpha>0$ and $\theta\in[0,1]$, we have
$u^*(\theta)>0$ for every type. Hence the proposed mechanism
is individually rational and therefore implementable.
Moreover,
$
    u^{*\prime}(\theta)
    =
    \frac{2}{2n+1}
    =
    x_{2h}^*(\theta)-x_{2l}^*(\theta)
$
for every $\theta\in[0,1]$.

We next establish global optimality, even when the participation
constraint is omitted. Let $x=\{x_1,x_{2h},x_{2l}\}$ be any
continuous incentive-compatible mechanism. Define its
interim expected payoff by
$
    u_x(\theta).
$
Because the action rules are continuous,
$x_{2h}-x_{2l}$ is continuous. The Envelope condition therefore
implies that $u_x\in C^1[0,1]$ and
$
    u_x'(\theta)
    =
    x_{2h}(\theta)-x_{2l}(\theta)$ for every $\theta\in[0,1]$.

Define the deviations from $x^*$ by
$
    d_1(\theta)
    \equiv
    x_1(\theta)-x_1^*(\theta),   d_h(\theta)
    \equiv
    x_{2h}(\theta)-x_{2h}^*(\theta),
    d_l(\theta)
    \equiv
    x_{2l}(\theta)-x_{2l}^*(\theta),
$
and define
$
    z(\theta)
    \equiv
    u_x(\theta)-u^*(\theta).
$
By the definitions of $u_x$ and $u^*$,
\begin{equation}
\label{eq:prop-beta-utility-deviation}
      z(\theta)
    =
    d_1(\theta)
    +\theta d_h(\theta)
    +(1-\theta)d_l(\theta).
\end{equation}
Because $u_x,u^*\in C^1[0,1]$, we have $z\in C^1[0,1]$.
Moreover, for every $\theta\in[0,1]$,
\begin{align}
   z'(\theta)
    =
    u_x'(\theta)-u^{*\prime}(\theta)
    &=
    \left[
        x_{2h}(\theta)-x_{2l}(\theta)
    \right]
    -
    \left[
        x_{2h}^*(\theta)-x_{2l}^*(\theta)
    \right]
    \notag\\
    &=
    d_h(\theta)-d_l(\theta).
    \label{eq:prop-beta-utility-derivative}
\end{align}

Now define the auxiliary function
$
    \lambda^*(\theta)
    \equiv
    -\frac{4}{2n+1}\theta(1-\theta)f(\theta)
    =
    -\frac{4}{(2n+1) B(\alpha,\beta)}
    \theta^\alpha(1-\theta)^\beta.
$
Because $\alpha,\beta>0$, the function $\lambda^*$ is absolutely
continuous on $[0,1]$ and satisfies
$
    \lambda^*(0)=\lambda^*(1)=0.
$
For every
$\theta\in(0,1)$, we have 
$
    \frac{d}{d\theta}
    \left[
        \theta(1-\theta)f(\theta)
    \right]
    =
    f(\theta)\bigl(\alpha-n\theta\bigr).
$
Consequently,
$
    \lambda^{*\prime}(\theta)
    =
    -\frac{4}{2n+1}
    f(\theta)\bigl(\alpha-n\theta\bigr).
$
A direct calculation gives, for every $\theta\in(0,1)$,

\begin{align}
    -2f(\theta)
    \bigl(x_1^*(\theta)-\theta\bigr)
    &=
    \lambda^{*\prime}(\theta),
    \label{eq:prop-beta-identity-x1}
    \\
    -2\theta f(\theta)
    \bigl(x_{2h}^*(\theta)-\theta\bigr)
    &=
    \theta\lambda^{*\prime}(\theta)
    +\lambda^*(\theta),
    \label{eq:prop-beta-identity-x2h}
    \\
    -2(1-\theta)f(\theta)
    \bigl(x_{2l}^*(\theta)-\theta\bigr)
    &=
    (1-\theta)\lambda^{*\prime}(\theta)
    -\lambda^*(\theta).
    \label{eq:prop-beta-identity-x2l}
\end{align}

The principal's expected payoff under a mechanism $x$ is
\[
    U_p(x)
    =
    -\int_0^1
    f(\theta)
    \left[
        (x_1(\theta)-\theta)^2
        +\theta(x_{2h}(\theta)-\theta)^2
        +(1-\theta)(x_{2l}(\theta)-\theta)^2
    \right]d\theta.
\]
All the terms below are integrable. Indeed, the action rules are
continuous and therefore bounded on $[0,1]$, the density $f$
integrates to one, and $\lambda^*$ is absolutely continuous.

Expanding the principal's quadratic payoff around $x^*$ and using
\eqref{eq:prop-beta-identity-x1}--%
\eqref{eq:prop-beta-identity-x2l}, we obtain
\begin{equation*}
\begin{aligned}
    &U_p(x)-U_p(x^*)
    ={}
    \int_0^1
    \Big[
        \lambda^{*\prime}(\theta)d_1(\theta)
        +
        \bigl(
            \theta\lambda^{*\prime}(\theta)
            +\lambda^*(\theta)
        \bigr)d_h(\theta)
    \\
    &
        +
        \bigl(
            (1-\theta)\lambda^{*\prime}(\theta)
            -\lambda^*(\theta)
        \bigr)d_l(\theta)
    \Big]d\theta
    -
    \int_0^1
    f(\theta)
    \left[
        d_1(\theta)^2
        +\theta d_h(\theta)^2
        +(1-\theta)d_l(\theta)^2
    \right]d\theta.
\end{aligned}
\end{equation*}

Using \eqref{eq:prop-beta-utility-deviation} and
\eqref{eq:prop-beta-utility-derivative}, the first integral in
$U_p(x)-U_p(x^*)$ can be rewritten as
\begin{align*}
    &\int_0^1
    \Big[
        \lambda^{*\prime}(\theta)
        \bigl(
            d_1(\theta)
            +\theta d_h(\theta)
            +(1-\theta)d_l(\theta)
        \bigr)
        +
        \lambda^*(\theta)
        \bigl(
            d_h(\theta)-d_l(\theta)
        \bigr)
    \Big]d\theta
    \\
    &\qquad=
    \int_0^1
    \left[
        \lambda^{*\prime}(\theta)z(\theta)
        +\lambda^*(\theta)z'(\theta)
    \right]d\theta =
    \left[
        \lambda^*(\theta)z(\theta)
    \right]_{\theta=0}^{\theta=1}
    =
    0.
\end{align*}
The second equality follows from integration by parts, which is valid
because $\lambda^*$ is absolutely continuous and $z\in C^1[0,1]$.
The final equality follows from
$
    \lambda^*(0)=\lambda^*(1)=0.
$

Therefore, we have
$
    U_p(x)-U_p(x^*)
    =
    -
    \int_0^1
    f(\theta)
    \left[
        d_1(\theta)^2
        +\theta d_h(\theta)^2
        +(1-\theta)d_l(\theta)^2
    \right]d\theta
    \leq0.
$
Thus, $x^*$ gives the principal at least as high an expected payoff as
any admissible incentive-compatible mechanism and is therefore
optimal. Its uniqueness comes from the fact that $f(\theta) >0$ on $(0,1)$ and that all allocation rules are continuous. Specifically, if $U_p(x) = U_p(x^*)$, then it must imply that $ d_1(\theta) =  d_h(\theta)  =  d_l(\theta) = 0$ for every $\theta\in(0,1)$, which means that $x = x^*$ on $(0,1)$. Continuity of the deviation functions then
implies the same equalities at $\theta=0$ and $\theta=1$. Hence,
$
    x=x^*
$
on all of $[0,1]$, proving uniqueness.

\end{proof}

\subsection{Proof of Lemma \ref{lem:existence}}

\begin{proof}
    \underline{    \textbf{\emph{Existence.}}}
Let us use $s(\theta) \equiv y$ as the effective screening type. Use the notations related to $y$ defined in the proof of Proposition~\ref{prop:nonnegative-spread}.  In particular, we have $g(y) \equiv f(\psi(y))/s'(\psi(y))$ denoting the density of $y$. Both $f$ and $s'$ are continuous, it follows that $g$ is continuous and that 
$
0<\underline g
\equiv \min_{y\in[\underline{y},\bar{y}]}g(y)
\le g(y)\le
\bar g
\equiv \max_{y\in[\underline{y},\bar{y}]}g(y)
<\infty.
$
Define $w(y) \equiv y(1-y)g(y)$. It is continuous and strictly positive in the interior of $[\underline{y},\bar{y}]$.

\medskip
\noindent
\emph{Step 1: Reduction to a convex-function problem.}

According to the proof of Proposition~\ref{prop:nonnegative-spread}, the principal's expected loss equals:
\[
\mathcal J(U)
\equiv
\int_{\underline{y}}^{\bar{y}}
\left[
\frac12\bigl(U(y)-2\psi(y)\bigr)^2g(y)
+w(y)\bigl(U'(y)\bigr)^2
\right]dy,
\]
where $U(y)\equiv u(\psi(y))$ is the agent's interim expected payoff as a function of the effective screening type $y$. $U$ is nonnegative, convex and continuously differentiable, with $U'(y)=\Delta(\psi(y))$. We temporarily drop the nonnegativity constraint and minimize
$\mathcal J$ over the larger class of all convex Lipschitz
functions on $[\underline y,\bar y]$. We then show that a
minimizer in this larger class is continuously differentiable
and nonnegative.

\medskip
\noindent
\emph{Step 2: Existence in the larger class.}

Choose a minimizing sequence $(U_n)$, that is, a sequence of
convex Lipschitz functions on $[\underline{y},\bar{y}]$ whose
losses converge to the infimum:
\[
\mathcal J(U_n)\longrightarrow
\inf\left\{
\mathcal J(U):
U\text{ is convex and Lipschitz on }[\underline{y},\bar{y}]
\right\}
\qquad\text{as }n\to\infty.
\]
The infimum is finite because $\mathcal J$ is nonnegative
and finite for constant functions. Consequently, the sequence
of losses is bounded, so there exists a finite constant $C>0$
such that $\mathcal J(U_n)\le C$ for every $n$. Replacing $U_n$ by $\max\{U_n,2\underline\theta\}$ preserves
convexity and Lipschitz continuity and weakly reduces both terms
in $\mathcal J$. Hence we may assume that
$U_n\ge 2\underline\theta$.

We next show that the slopes can be bounded uniformly.
Let
\[
a\equiv\frac{\underline{y}+\bar{y}}{2}\in(0,1),
\qquad
K_0\equiv
\left(
\frac{C}{
\min\left\{
\int_{\underline{y}}^{a}w(y)\,dy,\,
\int_a^{\bar{y}}w(y)\,dy
\right\}}
\right)^{1/2}.
\]
Convexity and the bound on the derivative term in $\mathcal J$
imply that both one-sided slopes of $U_n$ at $a$ lie in
$[-K_0,K_0]$. Indeed, a positive slope at $a$ bounds all slopes
to its right from below, while a negative slope bounds all
slopes to its left from above.

Choose $K>K_0$ sufficiently large that
$
2K>
\frac{2(\bar\theta-\underline\theta)\bar g}
{\underline g\min\{a,1-a\}}.
$
Writing $q_n=U_n'$ almost everywhere, define the capped slopes
and the corresponding functions by
$
\widehat q_n(y)
\equiv \max\{-K,\min\{q_n(y),K\}\},
\widehat U_n(y)
\equiv U_n(a)+\int_a^y\widehat q_n(t)\,dt.
$
The functions $\widehat U_n$ are convex and $K$-Lipschitz.
Because the slopes at $a$ lie strictly between $-K$ and $K$,
the caps only replace an outer tail by an affine continuation:
the slope is $-K$ on an altered left tail and $K$ on an altered
right tail. Consequently,
$
2\underline\theta\le \widehat U_n(y)\le U_n(y)
\;\text{for every }y\in[\underline{y},\bar{y}].
$

Fubini's theorem gives
\begin{align*}
&\int_{\underline{y}}^{\bar{y}}
\bigl(U_n(y)-\widehat U_n(y)\bigr)g(y)\,dy \\
&\quad =
\int_{\underline{y}}^{a}
\bigl|q_n(t)-\widehat q_n(t)\bigr|
\left(\int_{\underline{y}}^{t}g(y)\,dy\right)dt +
\int_a^{\bar{y}}
\bigl|q_n(t)-\widehat q_n(t)\bigr|
\left(\int_t^{\bar{y}}g(y)\,dy\right)dt \\
&\quad \le
\frac{\bar g}{\underline g\min\{a,1-a\}}
\int_{\underline{y}}^{\bar{y}}
w(t)\bigl|q_n(t)-\widehat q_n(t)\bigr|\,dt.
\end{align*}
For the last inequality, use
$\underline{y}\ge0$, $\bar{y}\le1$, and the bounds on $g$. Since $2\psi(y)\le 2\bar\theta$ and $\widehat U_n(y)\ge 2\underline\theta$,
\begin{align*}
\frac12\left[
(\widehat U_n-2\psi)^2-(U_n-2\psi)^2
\right] & =
\left(2\psi-\frac{U_n+\widehat U_n}{2}\right)
\bigl(U_n-\widehat U_n\bigr) \le
2(\bar\theta-\underline\theta)
\bigl(U_n-\widehat U_n\bigr).
\end{align*}
Moreover,
$
q_n^2-\widehat {q_n}^2
\ge 2K\bigl|q_n-\widehat q_n\bigr|.
$
Combining these inequalities yields
\begin{align*}
\mathcal J(\widehat U_n)-\mathcal J(U_n)
&\le
\left(
\frac{2(\bar\theta-\underline\theta)\bar g}
{\underline g\min\{a,1-a\}}
-2K
\right)
\int_{\underline{y}}^{\bar{y}}
w(y)\bigl|q_n(y)-\widehat q_n(y)\bigr|\,dy \le 0.
\end{align*}
Thus $(\widehat U_n)$ remains a minimizing sequence and has a
common Lipschitz constant.

The loss bound also gives a uniform bound on its levels.
Indeed, since $\int_{\underline{y}}^{\bar{y}}g(y)\,dy=1$,
$
\|\widehat U_n\|_\infty
\le
\max\{|2\underline\theta|,|2\bar\theta|\}+\sqrt{2C}
+K(\bar{y}-\underline{y}).
$
By the Arzel\`a--Ascoli theorem, a subsequence converges
uniformly to a convex $K$-Lipschitz function $U^*$.
Convexity implies that
$
\widehat U_n'(y)\longrightarrow U^{*\prime}(y)
\; \text{for almost every }y\in[\underline{y},\bar{y}].
$
The uniform bounds on levels and slopes therefore allow
dominated convergence, giving
$
\mathcal J(U^*)
=
\lim_{n\to\infty}\mathcal J(\widehat U_n)
=
\inf\left\{
\mathcal J(U):
U\text{ is convex and Lipschitz on }[\underline{y},\bar{y}]
\right\}.
$

\medskip
\noindent
\emph{Step 3: Continuity of the induced allocation rules.}

Suppose that $U^*$ has a kink at an interior point $r$, so that
$
U^{*\prime}_-(r)
<
U^{*\prime}_+(r).
$
For sufficiently small $\varepsilon>0$, replace $U^*$ on
$[r-\varepsilon,r+\varepsilon]$ by the chord joining its values
at the two endpoints, leaving it unchanged elsewhere.
Denote the resulting function by $V_\varepsilon$.
This function remains convex and Lipschitz.

The modification changes the level-loss term by
$O(\varepsilon^2)$: its support has length $2\varepsilon$,
and $\|V_\varepsilon-U^*\|_\infty=O(\varepsilon)$.
The chord's slope converges to
$
\frac{U^{*\prime}_-(r)+U^{*\prime}_+(r)}{2},
$
while the slopes of $U^*$ on the two sides converge to
$U^{*\prime}_-(r)$ and $U^{*\prime}_+(r)$, respectively.
By continuity of $w$,
\[
\mathcal J(V_\varepsilon)-\mathcal J(U^*)
=
-\frac{\varepsilon}{2}w(r)
\left(
U^{*\prime}_+(r)-U^{*\prime}_-(r)
\right)^2
+o(\varepsilon)
<0
\]
for sufficiently small $\varepsilon$, contradicting optimality.
Hence $U^*$ has no interior kinks.

Its derivative is therefore continuous in the interior.
Because this derivative is increasing and bounded, it has
finite limits at both endpoints, equal to the corresponding
one-sided derivatives. Thus
$U^*\in C^1([\underline{y},\bar{y}])$.

Finally, define
$
u^*(\theta)\equiv U^*(s(\theta)),
\;
\Delta^*(\theta)\equiv U^{*\prime}(s(\theta)),
$
and set
$
x_1^*(\theta)=\frac{u^*(\theta)}{2},
\;
x_{2h}^*(\theta)
=\frac{u^*(\theta)}{2}
 +(1-s(\theta))\Delta^*(\theta),
\;
x_{2l}^*(\theta)
=\frac{u^*(\theta)}{2}
 -s(\theta)\Delta^*(\theta).
$
The bound $\widehat U_n(y)\ge2\underline\theta$ passes to the
uniform limit, giving $u^*(\theta)=U^*(s(\theta)) \ge 2\theta \ge0$ for every $\theta\in\Theta$.

The induced mechanism therefore satisfies both incentive
compatibility and individual rationality.
Its expected loss equals the minimum of the relaxed problem,
which is a lower bound on the loss of every implementable
mechanism. Hence it is optimal for the principal's problem.

\medskip
\noindent
\underline{\textbf{\emph{Uniqueness.}}}
Suppose that $x=(x_1,x_{2h},x_{2l})$ and
$\widehat x=(\widehat x_1,\widehat x_{2h},\widehat x_{2l})$
are two optimal mechanisms. Their pointwise average
$x^m=(x+\widehat x)/2$ is continuous and feasible, since
both the IC and IR
constraints are linear in the allocation rules.

By the quadratic form of the principal's payoff,
\begin{align*}
U_p(x^m)
=
\frac{U_p(x)+U_p(\widehat x)}{2}+ &
\frac14\int_\Theta
\Big[
\bigl(x_1(\theta)-\widehat x_1(\theta)\bigr)^2
+s(\theta)
\bigl(x_{2h}(\theta)-\widehat x_{2h}(\theta)\bigr)^2\\
&\qquad\qquad
+(1-s(\theta))
\bigl(x_{2l}(\theta)-\widehat x_{2l}(\theta)\bigr)^2
\Big]\,dF(\theta).
\end{align*}
Since $x$ and $\widehat x$ are optimal and $x^m$ is feasible,
the nonnegative integral must equal zero.
Because $f(\theta)>0$ and $0<s(\theta)<1$ for every
$\theta\in\operatorname{int}(\Theta)$, this implies
$x=\widehat x$ almost everywhere on $\Theta$.
Continuity then implies $x=\widehat x$ everywhere on $\Theta$.
Hence the optimal mechanism is unique.

\end{proof}

\subsection{Proof of Theorem \ref{thm:frontloading-information}}
\begin{proof}

The proof of Proposition~\ref{prop:nonnegative-spread} shows that $m^*(\theta)=x_1^*(\theta)$ for every $\theta\in\Theta$. For each $\theta$, the conditional variance decomposition gives
\begin{align*}
    s(\theta)
    \bigl(x_{2h}^*(\theta)-\theta\bigr)^2
    +
    \bigl(1-s(\theta)\bigr)
    \bigl(x_{2l}^*(\theta)-\theta\bigr)^2
    &=
    \bigl(m^*(\theta)-\theta\bigr)^2
    +
    s(\theta)\bigl(1-s(\theta)\bigr)
    \bigl(\Delta^*(\theta)\bigr)^2
    \\
    &=
    \bigl(x_1^*(\theta)-\theta\bigr)^2
    +
    s(\theta)\bigl(1-s(\theta)\bigr)
    \bigl(\Delta^*(\theta)\bigr)^2.
\end{align*}
Integrating over $\theta$ yields
$
    L_2^*
    =
    L_1^*
    +
    \int_{\Theta}
    s(\theta)\bigl(1-s(\theta)\bigr)
    \bigl(\Delta^*(\theta)\bigr)^2
    \,dF(\theta)
    \geq
    L_1^*.
$
Therefore,
$
    2L_1^*
    \leq
    L_1^*+L_2^*,
$ which proves that the optimal mechanism frontloads information value.
\end{proof}

\subsection{Proposition \ref{prop:generallossfunctions} and its proof}
\begin{proposition}
    \label{prop:generallossfunctions}
  The following statements hold:
  \begin{enumerate}[(i)]
      \item When the loss function $l$ is continuously differentiable and strictly convex, define
$\Gamma(z)\equiv l\bigl((l')^{-1}(z)\bigr),$
which gives the loss associated with an allocation error whose marginal loss is $z$. If $\Gamma$ is convex, then the optimal mechanism frontloads information value.
      \item If $l(e)=h(|e|)$, where $h$ is increasing and concave with $h(0)=0$, then there exists an optimal mechanism that frontloads information value.
  \end{enumerate}
\end{proposition}

\begin{proof}
    For part (i), fix a type $\theta$, and let $e_1 \equiv x_1 - \theta$, $e_{2h} \equiv x_{2h} - \theta$, and $e_{2l} \equiv x_{2l} - \theta$ denote the allocation errors in the first and second periods. Let $z_h \equiv l'(e_{2h})$ and $z_l \equiv l'(e_{2l})$ denote the corresponding second-period marginal losses. Optimality requires expected marginal loss to be equalized across the two periods, which then implies
    $
    l(e_1) = \Gamma(sz_h + (1-s)z_l).
    $
     The expected second-period loss is    
     $
    s l(e_{2h}) + (1-s) l(e_{2l})
    = s \Gamma(z_h) + (1-s) \Gamma(z_l).
    $
    Hence, 
$
    \mathbb{E}[l(e_2)|\theta] - l(e_1) = s \Gamma(z_h) + (1-s) \Gamma(z_l) - \Gamma(sz_h + (1-s)z_l) \geq 0,
$
  where the inequality follows from Jensen's inequality because $\Gamma$ is convex. Thus, for every type $\theta$, the expected second-period loss is weakly larger than the first-period loss. Integrating over types gives $L_1^* \leq L_2^*$. The optimal mechanism therefore frontloads information value.

    For part (ii), let $a,b>0$. Concavity of $h$ and $h(0)=0$ imply
    \[
    h(a) = h \left(\frac{a}{a+b} (a+b) + \frac{b}{a+b} 0\right) \geq \frac{a}{a+b} h(a+b).
    \]
    Similarly, $h(b) \geq \frac{b}{a+b} h(a+b)$. Adding the two inequalities yields $h(a) + h(b) \geq h(a+b)$. Thus, for arbitrary signed errors $e_1$ and $e_2$,
    \[
    l(e_1+e_2) = h(|e_1+e_2|) \leq h(|e_1| + |e_2|) \leq h(|e_1|) + h(|e_2|) = l(e_1) + l(e_2),
    \]
   where the first inequality follows from the triangle inequality and the monotonicity of $h$. Hence, concentrating two allocation errors in a single period is less costly than incurring them separately across the two periods. Now consider any implementable mechanism $(x_1,x_{2h},x_{2l})$. Construct a new mechanism $(\tilde x_1,\tilde x_{2h},\tilde x_{2l})$ by setting, for every report $\hat\theta$,
$
\tilde x_1(\hat\theta)=\hat\theta,
\;
\tilde x_{2h}(\hat\theta)
=
x_1(\hat\theta)+x_{2h}(\hat\theta)-\hat\theta,
\;
\tilde x_{2l}(\hat\theta)
=
x_1(\hat\theta)+x_{2l}(\hat\theta)-\hat\theta.
$
For every report and every signal realization, the agent's total allocation is unchanged; the transformation therefore preserves both incentive compatibility and individual rationality.
For a truthful type $\theta$, the first-period error under the new mechanism is zero, while the second-period error following signal $\sigma$ equals the sum of the two original errors. By the inequality above, the transformation reduces total loss for every type and every signal realization, and hence weakly increases the principal's expected payoff. Therefore, we can apply this transformation to any optimal mechanism to obtain another optimal mechanism with zero first-period error, which frontloads information value.
\end{proof}

\subsection{Proof of Remark \ref{rmk:generallossfunctions}}

\begin{proof}
    This result is a direct consequence of Proposition~\ref{prop:generallossfunctions}. When $p>1$, this class falls under case (i) of Proposition~\ref{prop:generallossfunctions}, whereas for $0<p\leq 1$, it falls under case (ii).
\end{proof}

\subsection{Proof of Proposition \ref{prop:nonnegative-spread}}

\begin{proof}

For any mechanism, define its expected second-period action by
$
    m(\theta)
    \equiv
    s(\theta)x_{2h}(\theta)
    +
    \bigl(1-s(\theta)\bigr)x_{2l}(\theta),
$
and the agent's truthful interim payoff by $ u(\theta)
    \equiv
    x_1(\theta)+m(\theta).$
The two second-period actions can be recovered from $m$ and $\Delta$
as
$
    x_{2h}(\theta)
    =
    m(\theta)
    +
    \bigl(1-s(\theta)\bigr)\Delta(\theta), \;    x_{2l}(\theta)
    =
    m(\theta)
    -
    s(\theta)\Delta(\theta).
$

An agent with true type $\theta$ who reports $\hat{\theta}$ obtains
expected payoff
\begin{align}
    &x_1(\hat{\theta})
    +s(\theta)x_{2h}(\hat{\theta})
    +\bigl(1-s(\theta)\bigr)x_{2l}(\hat{\theta})
    \notag =
    u(\hat{\theta})
    +
    \bigl(
        s(\theta)-s(\hat{\theta})
    \bigr)\Delta(\hat{\theta}).
    \label{eq:lemma2-report-payoff}
\end{align}
Thus, all reporting incentives depend on the mechanism only through
the pair $(u,\Delta)$.

We first derive a balancing property of the optimal mechanism.
Conditional on type $\theta$, the principal's quadratic loss is
\begin{align*}
    &\bigl(x_1(\theta)-\theta\bigr)^2
    +s(\theta)
        \bigl(x_{2h}(\theta)-\theta\bigr)^2
    +\bigl(1-s(\theta)\bigr)
        \bigl(x_{2l}(\theta)-\theta\bigr)^2
    \\
    &\qquad=
    \bigl(x_1(\theta)-\theta\bigr)^2
    +\bigl(m(\theta)-\theta\bigr)^2
    +s(\theta)\bigl(1-s(\theta)\bigr)
        \Delta(\theta)^2
    \\
    &\qquad=
    \frac{1}{2}
    \bigl(u(\theta)-2\theta\bigr)^2
    +
    \frac{1}{2}
    \bigl(x_1(\theta)-m(\theta)\bigr)^2
    +
    s(\theta)\bigl(1-s(\theta)\bigr)
        \Delta(\theta)^2.
\end{align*}
For fixed $u(\theta)$ and $\Delta(\theta)$, this expression is
uniquely minimized by
$
    x_1(\theta)=m(\theta)=\frac{u(\theta)}{2}.
$

Replacing $x_1$ and $m$ by $u/2$, while holding $u$ and $\Delta$
fixed, preserves all reporting incentives. Moreover, because $u$, $\Delta$, and
$s$ are continuous, the resulting action rules are continuous.
Therefore, optimality requires 
$
    x_1^*(\theta)
    =
    m^*(\theta)
    =
    \frac{u^*(\theta)}{2}
    \;
    \text{for every }\theta\in\Theta.
$
Indeed, if $x_1^*(\theta_0)\neq m^*(\theta_0)$ at some type
$\theta_0$, continuity would imply that
$x_1^*-m^*$ is bounded away from zero on a neighborhood of
$\theta_0$. Since $f$ is strictly positive, the
balanced mechanism would then strictly lower the principal's
expected loss, contradicting optimality.

We next use the signal probability as the effective screening type.
Let $    \underline{y} \equiv s(\underline{\theta}),
    \bar{y} \equiv s(\bar{\theta}),$
and define
$
    \psi(y)\equiv s^{-1}(y),
    \;
    y\in[\underline{y},\bar{y}].
$
Because $s\in C^1(\Theta)$ and $s'(\theta)>0$, the inverse $\psi$ is
well defined and strictly increasing. The density of
$y=s(\theta)$ is
$
    g(y)
    = \frac{
        f(\psi(y))
    }{
        s'(\psi(y))
    },
$
which is strictly positive on $[\underline{y} ,\bar{y}]$.

Define
$
    U^*(y)
    \equiv
    u^*(\psi(y)), \;  \widetilde{\Delta}^*(y)
    \equiv
    \Delta^*(\psi(y)).
$
Since the action rules are continuous,
$\widetilde{\Delta}^*$ is continuous. Incentive
compatibility is equivalent to
\begin{equation}
\label{eq:lemma2-ic-y}
    U^*(y)
    \geq
    U^*(\hat{y})
    +(y-\hat{y})
        \widetilde{\Delta}^*(\hat{y})
    \qquad
    \text{for all }y,\hat{y}\in[\underline{y},\bar{y}].
\end{equation}
For any $y_2>y_1$, applying
\eqref{eq:lemma2-ic-y} in both directions gives
$
    \widetilde{\Delta}^*(y_1)
    \leq
    \frac{
        U^*(y_2)-U^*(y_1)
    }{
        y_2-y_1
    }
    \leq
    \widetilde{\Delta}^*(y_2).
$
It follows that $U^*$ is convex and
$\widetilde{\Delta}^*$ is increasing. Since
$\widetilde{\Delta}^*$ is continuous, letting $y_2$ converge to
$y_1$ yields
$
    {U^*}'(y)
    =
    \widetilde{\Delta}^*(y)
    \;
    \text{for every }y\in[\underline{y},\bar{y}],
$
where the derivatives at the endpoints are understood as one-sided
derivatives.

Using the balancing condition $ x_1^*(\theta)
    =
    m^*(\theta)$ and changing
variables from $\theta$ to $y$, the principal's expected loss can be
written as
\begin{equation}
\label{eq:lemma2-reduced-objective}
    \mathcal{J}(U)
    =
    \int_{\underline{y}}^{\bar{y}}
    \left[
        \frac{1}{2}
        \bigl(
            U(y)-2\psi(y)
        \bigr)^2
        +
        y(1-y)\bigl(U'(y)\bigr)^2
    \right]
    g(y)\,dy.
\end{equation}
The relevant feasible functions are nonnegative, continuously
differentiable, and convex. Conversely, every such function $U$
can be implemented by a continuous mechanism satisfying both
incentive compatibility and individual rationality. Indeed,
letting
$
    u(\theta)=U(s(\theta)),
    \;
    \Delta(\theta)=U'(s(\theta)),
$
and setting
\[
    x_1(\theta)=\frac{u(\theta)}{2},  \quad  x_{2h}(\theta)
    =
    \frac{u(\theta)}{2}
    +
    \bigl(1-s(\theta)\bigr)\Delta(\theta), \quad    x_{2l}(\theta)
    =
    \frac{u(\theta)}{2}
    -
    s(\theta)\Delta(\theta)
\]
gives continuous action rules. Incentive compatibility follows from
the convexity inequality
$
    U(y)
    \geq
    U(\hat{y})
    +(y-\hat{y})U'(\hat{y}).
$

Suppose, toward a contradiction, that
$
    \Delta^*(\theta_0)<0
$
for some $\theta_0\in\Theta$. By continuity,
$\Delta^*$ is strictly negative on a nondegenerate interval.
Equivalently, by $ {U^*}'(y)
    =
    \widetilde{\Delta}^*(y)$,
${U^*}'$ is strictly negative on a set of positive measure.
Therefore, $U^*$ is not increasing.

Let
$
    c
    \equiv
    \min
    \left(
        \operatorname*{arg\,min}_{y\in[\underline{y},\bar{y}]}U^*(y)
    \right)
$
be the leftmost minimizer of $U^*$. Since $U^*$ is convex, it is
decreasing on $[\underline{y},c]$ and increasing on $[c,\bar{y}]$. Since $U^*$
is not increasing, we have $c>\underline{y}$.

Define
$
    R(y)\equiv U^*(y)-2\psi(y).
$
Because $U^*$ is decreasing and $2\psi(y)$ is strictly increasing on
$[\underline{y},c]$, the function $R$ is strictly decreasing on this interval.
Choose a constant $k$ as follows:
\[
    k
    \equiv
    \begin{cases}
        U^*(c),
        &
        \text{if }R(c)\geq0,
        \\[2mm]
        U^*(\underline{y}),
        &
        \text{if }R(\underline{y})\leq0,
        \\[2mm]
        U^*(y_0)=2\psi(y_0),
        &
        \text{if }R(\underline{y})>0>R(c),
    \end{cases}
\]
where, in the third case, $y_0\in(\underline{y},c)$ is the point satisfying $U^*(y_0)=2\psi(y_0)$.
Let
$
    r
    \equiv
    \sup
    \left\{
        y\in[c,\bar{y}]:
        U^*(y)\leq k
    \right\}.
$
The set is nonempty because it contains $c$. Define
\[
    \widehat{U}(y)
    \equiv
    \begin{cases}
        k,
        &
        y\in[\underline{y},r],
        \\[1mm]
        U^*(y),
        &
        y\in(r,\bar{y}].
    \end{cases}
\]
If $r<\bar{y}$, continuity implies $U^*(r)=k$. Since $r\geq c$ and $U^*$
is increasing on $[c,\bar{y}]$, the function $\widehat{U}$ is convex and
increasing. Because $k$ is a value of the nonnegative function $U^*$,
we have $k\ge0$. Hence $\widehat U\ge0$, so the flattening
also preserves individual rationality.

By construction, whenever $\widehat{U}(y)\neq U^*(y)$, the constant
$k$ lies between $U^*(y)$ and $2\psi(y)$. Specifically, if
$R(c)\geq0$, then
$
    2\psi(y)\leq k\leq U^*(y);
$
if $R(\underline{y})\leq0$, then
$
    U^*(y)\leq k\leq 2\psi(y);
$
and, in the remaining case,
$
    2\psi(y)\leq k\leq U^*(y)
    \;\text{for }y\leq y_0,
$
while
$
    U^*(y)\leq k\leq 2\psi(y)
    \;\text{for }y\geq y_0.
$
Consequently,
\begin{equation*}
    \bigl(
        \widehat{U}(y)-2\psi(y)
    \bigr)^2
    \leq
    \bigl(
        U^*(y)-2\psi(y)
    \bigr)^2
    \qquad
    \text{for every }y\in[\underline{y},\bar{y}].
\end{equation*}
Moreover,
$
    \widehat{U}'(y)
    =
    0
    \;
    \text{for almost every }y\in(\underline{y},r),
$
and
$
    \widehat{U}'(y)
    =
    {U^*}'(y)
    \;
    \text{for almost every }y\in(r,\bar{y}).
$
Thus,
\begin{equation*}
    y(1-y)
    \bigl(
        \widehat{U}'(y)
    \bigr)^2
    \leq
    y(1-y)
    \bigl(
        {U^*}'(y)
    \bigr)^2
\end{equation*}
almost everywhere. The inequality is strict on a set of positive
measure because ${U^*}'<0$ on such a set and
$
    y(1-y)g(y)>0
$
almost everywhere in the interior of $[\underline{y},\bar{y}]$. Therefore, we have
\begin{equation}
\label{eq:lemma2-strict-improvement}
    \mathcal{J}(\widehat{U})
    <
    \mathcal{J}(U^*).
\end{equation}

The function $\widehat{U}$ may have a kink at $r$, so it need not
directly induce continuous action rules. We therefore approximate it
by continuously differentiable convex functions.

Let $\widehat{U}'_+(y)$ denote the right derivative of
$\widehat{U}$ for $y\in[\underline{y},\bar{y})$.
Since $\widehat{U}$ is convex and increasing, this derivative
is nonnegative and increasing. Moreover, by construction,
$\widehat{U}'_+(y)$ equals either zero or ${U^*}'(y)$.
Since $U^*\in C^1([\underline{y},\bar{y}])$, it follows that
$
    0\leq \widehat{U}'_+(y)\leq M
    \equiv
    \max_{z\in[\underline{y},\bar{y}]}
    \bigl|{U^*}'(z)\bigr|
    <\infty.
$
Define its value at the upper endpoint by
$
    \widehat{U}'_+(\bar{y})
    \equiv
    \lim_{y\uparrow\bar{y}}\widehat{U}'_+(y),
$
where the limit exists and is finite because
$\widehat{U}'_+$ is increasing and bounded.
Extend $\widehat{U}'_+$ to $\mathbb{R}$ by setting it equal
to $\widehat{U}'_+(\underline{y})$ for $y<\underline{y}$
and to $\widehat{U}'_+(\bar{y})$ for $y>\bar{y}$. Let $\rho:\mathbb{R}\to\mathbb{R}_+$ be a smooth function with compact support satisfying
$
    \int_{\mathbb{R}}\rho(t)\,dt=1,
$
and, for each $j\in\mathbb{N}$, define
$
p_j(y) \equiv \int_{\mathbb{R}} j\rho(jt)  \widehat{U}'_+(y-t)\,dt.
$

$p_j(y)$ is a weighted average of the slopes near $y$. It is smooth, nonnegative, and increasing. Define
$
    U_j(y)
    \equiv
    \widehat{U}(\underline{y})
    +
    \int_{\underline{y}}^y p_j(t)\,dt \ge k \ge 0.
$
Then $U_j$ is $C^1$, nonnegative, convex, and increasing. Moreover,
$
    U_j\longrightarrow\widehat{U}
    \;\text{uniformly on }[\underline{y},\bar{y}],
$
and
$
    U_j'(y)=p_j(y)
    \longrightarrow
    \widehat{U}'(y)
$
for almost every $y\in[\underline{y},\bar{y}]$. Since the sequence of derivatives is
uniformly bounded, dominated convergence applied to
\eqref{eq:lemma2-reduced-objective} gives
$
    \mathcal{J}(U_j)
    \longrightarrow
    \mathcal{J}(\widehat{U}).
$
By \eqref{eq:lemma2-strict-improvement}, for all sufficiently large
$j$, $    \mathcal{J}(U_j)
    <
    \mathcal{J}(U^*).$
For such a $j$, define $    u_j(\theta)\equiv U_j(s(\theta)),
    \; \Delta_j(\theta)\equiv U_j'(s(\theta)),$ 
and
\[
    x_{1,j}(\theta)
    =
    \frac{u_j(\theta)}{2}, \; \; x_{2h,j}(\theta)
    =
    \frac{u_j(\theta)}{2}
    +
    \bigl(1-s(\theta)\bigr)\Delta_j(\theta), \; \; x_{2l,j}(\theta)
    =
    \frac{u_j(\theta)}{2}
    -
    s(\theta)\Delta_j(\theta).
\]

These action rules are continuous. Since $U_j$ is convex,
$
    U_j(y)
    \geq
    U_j(\hat{y})
    +(y-\hat{y})U_j'(\hat{y}),
$
so the induced mechanism is incentive compatible. Individual rationality follows from $u_j(\theta) = U_j(s(\theta)) \geq k \geq 0$. Its expected loss
is $\mathcal{J}(U_j)$, which is strictly lower than
$\mathcal{J}(U^*)$. This contradicts the optimality of $x^*$.

Therefore, no $\theta_0\in\Theta$ can satisfy
$\Delta^*(\theta_0)<0$. Hence,
$
    \Delta^*(\theta)\geq0
    \;
    \text{for every }\theta\in\Theta.
$
\end{proof}

\subsection{Proof of Theorem \ref{thm:general-monotonicity}}

\begin{proof}

The proof of Proposition~\ref{prop:nonnegative-spread} shows that
$
    u^*(\theta)=2x_1^*(\theta).
$
Combining this equality with $ u'(\theta)
    =
    s'(\theta)\Delta(\theta)$ (from incentive compatibility) yields
${x_{1}^*}'(\theta) = \frac{1}{2}s'(\theta)\Delta^*(\theta)$. Since 
$s'(\theta)>0$ and $\Delta^*(\theta)\geq0$, the derivative of
$x_1^*$ is nonnegative. Therefore, $x_1^*$ is increasing.

We next establish the monotonicity of $x_{2l}^*$ without requiring
$\Delta^*$ to be differentiable. From $x_1^*(\theta) = m^*(\theta)$, we have $x_{2l}^*(\theta) = x_1^*(\theta) - s(\theta)\Delta^*(\theta)$. Therefore,
$
    x_{2l}^*(\theta_i)
    =
    \frac{u_i}{2}-s_i\Delta_i
$ for $ i\in\{1,2\}$.
Using the upper bound in \eqref{eq:lemma1-general-IC-bounds}, we obtain
\begin{align*}
    x_{2l}^*(\theta_2)-x_{2l}^*(\theta_1)
    &=
    \frac{1}{2}(u_2-u_1)
    -s_2\Delta_2+s_1\Delta_1
    \\
    &\leq
    \frac{1}{2}(s_2-s_1)\Delta_2
    -s_2\Delta_2+s_1\Delta_1
    \\
    &=
    -\frac{1}{2}
    (s_2-s_1)\Delta_2
    -
    s_1(\Delta_2-\Delta_1)
    \leq0.
\end{align*}
The final inequality follows from
$
    s_2\geq s_1\geq0,
    \;
    \Delta_2\geq\Delta_1,
    \;
    \Delta_2\geq0.
$
Hence $x_{2l}^*$ is decreasing on $\Theta$.

Finally, consider any interior type at which $\Delta^*$ is
differentiable. Differentiating
$x_{2l}^*(\theta) = x_1^*(\theta) - s(\theta)\Delta^*(\theta)$ and using ${x_{1}^*}'(\theta) = \frac{1}{2}s'(\theta)\Delta^*(\theta)$ gives
\begin{align*}
    {x_{2l}^*}'(\theta)
    &=
    \frac{1}{2}s'(\theta)\Delta^*(\theta)
    -
    s'(\theta)\Delta^*(\theta)
    -
    s(\theta){\Delta^*}'(\theta)
    =
    -\frac{1}{2}s'(\theta)\Delta^*(\theta)
    -
    s(\theta){\Delta^*}'(\theta).
\end{align*}
Because $\Delta^*$ is
increasing,
$
    {\Delta^*}'(\theta)\geq0
$
at every point at which its derivative exists.

Similarly, differentiating 
$ x_{2h}^*(\theta)
    =
    x_1^*(\theta)
    +
    \bigl(1-s(\theta)\bigr)\Delta^*(\theta)$, we obtain 
\begin{align*}
    {x_{2h}^*}'(\theta)
    &=
    \frac{1}{2}s'(\theta)\Delta^*(\theta)
    -
    s'(\theta)\Delta^*(\theta)
    +
    \bigl(1-s(\theta)\bigr)
        {\Delta^*}'(\theta)
    =
    -\frac{1}{2}s'(\theta)\Delta^*(\theta)
    +
    \bigl(1-s(\theta)\bigr)
        {\Delta^*}'(\theta).
\end{align*}

It remains to establish the decreasing property of
$x_{2h}^*$ near the highest type. Let us use $s(\theta) \equiv y$ as the effective screening type. Use the notations related to $y$ defined in the proof of Proposition~\ref{prop:nonnegative-spread}. 
The envelope and balancing conditions imply
$
    {U^*}'(y)=\Delta^*(\psi(y))
$
and
\begin{equation}
\label{eq:terminal-x2h-representation}
    x_{2h}^*(\psi(y))
    =
    \frac{U^*(y)}{2}
    +
    (1-y)\Delta^*(\psi(y)).
\end{equation}
By incentive compatibility and the nonnegative-spread result,
$\Delta^*(\psi(y))$ is continuous, nonnegative, and increasing in $y$.

We first note that
$
    \Delta^*(\psi(y))\not\equiv0.
$
Indeed, if $\Delta^*(\psi(y))\equiv0$, then
$ {U^*}'(y)=\Delta^*(\psi(y))$ implies that $U^*$ is constant.
Let $R_1$ and $R_2$ denote the expected losses in the first and second periods under the static optimal mechanism.
The balancing condition and
\eqref{eq:terminal-x2h-representation} would then imply that all three
action rules coincide with the same constant. The smallest total
expected loss attainable by such a mechanism is $2R_1$.
The static optimal mechanism, however, is feasible and has total
expected loss
$
    R_1+R_2<2R_1,
$
where the strict inequality follows from the informativeness of the
signal. This contradicts optimality and proves $\Delta^*(\psi(y))\not\equiv0$.

We distinguish two cases.

\medskip
\noindent
\emph{Case 1: $\bar{y}<1$.}

As shown in the proof of
Proposition~\ref{prop:nonnegative-spread}, the reduced loss is
\begin{equation*}
    \mathcal J(U)
    =
    \int_{\underline{y}}^{\bar{y}}
    \left[
        \frac12\bigl(U(y)-2\psi(y)\bigr)^2
        +y(1-y)\bigl(U'(y)\bigr)^2
    \right]g(y)\,dy,
\end{equation*}
and every continuously differentiable, convex, increasing
function $U$ induces a continuous incentive-compatible mechanism.

Define
$
    M
    \equiv
    \max_{y\in[\underline{y},\bar{y}]}
    \left|U^*(y)-2\psi(y)\right|.
$
Since $\Delta^*(\psi(\bar{y}))>0$ and $0<\bar{y}<1$,
we can choose $c\in(\underline{y},\bar{y})$ sufficiently close
to $\bar{y}$ that
$
    2y(1-y)\Delta^*(\psi(y))
    >
    M(\bar{y}-c)
    \;
    \text{for every }y\in[c,\bar{y}].
$
Define
\[
    \widehat U(y)
    =
    \begin{cases}
        U^*(y), & y\leq c,\\
        U^*(c)+\Delta^*(\psi(c))(y-c), & y>c.
    \end{cases}
\]
The function $\widehat U$ is continuously differentiable, nonnegative, convex,
and increasing. Hence
$
    U_\varepsilon
    \equiv
    (1-\varepsilon)U^*+\varepsilon\widehat U,
    \; \varepsilon\in[0,1],
$
also induces a continuous incentive compatible and individually rational mechanism.

For $y\in[c,\bar{y}]$, let
$
    h(y)
    \equiv
    \Delta^*(\psi(y))-\Delta^*(\psi(c))
    \geq0,
    \;
    D(y)
    \equiv
    U^*(y)-\widehat U(y)
    =
    \int_c^y h(t)\,dt.
$
Since $h$ is increasing,
$
    0\leq D(y)\leq(\bar{y}-c)h(y).
$
Optimality of $U^*$ therefore implies

\begin{align*}
    0
    &\leq
    \left.
        \frac{d}{d\varepsilon}\mathcal J(U_\varepsilon)
    \right|_{\varepsilon=0^+}
    =
    \int_c^{\bar{y}}
    \left[
        -\bigl(U^*(y)-2\psi(y)\bigr)D(y)
        -2y(1-y)\Delta^*(\psi(y))h(y)
    \right]g(y)\,dy
    \\
    &\leq
    \int_c^{\bar{y}}
    \left[
        M(\bar{y}-c)
        -2y(1-y)\Delta^*(\psi(y))
    \right]h(y)g(y)\,dy.
\end{align*}
The expression in square brackets is strictly negative throughout
$[c,\bar{y}]$. Since $h$ is continuous and nonnegative, and
$g>0$ almost everywhere, the preceding inequality requires
$h\equiv0$. Thus $\Delta^*(\psi(y))$ is constant and strictly
positive on $[c,\bar{y}]$.

By $ {U^*}'(y)=\Delta^*(\psi(y))$ and
\eqref{eq:terminal-x2h-representation},
$
    \frac{d}{dy}x_{2h}^*(\psi(y))
    =
    -\frac12\Delta^*(\psi(c))
    <0
    \;
    \text{for }y\in(c,\bar{y}).
$
Hence $x_{2h}^*(\psi(y))$ is strictly decreasing on
$[c,\bar{y}]$, even without the Lipschitz assumption.

\medskip
\noindent
\emph{Case 2: $\bar{y}=1$.}

First, integrating the envelope condition gives
$U^*(1)-U^*(y)
    =
    \int_y^1\Delta^*(\psi(t))\,dt.$
At $y=1$, the balancing condition gives
$x_{2h}^*(\psi(1))=U^*(1)/2$. Hence, for every
$y\in[\underline{y},1)$,
\begin{align*}
    \frac{
        x_{2h}^*(\psi(y))-x_{2h}^*(\psi(1))
    }{1-y}
    &=
    \frac{U^*(y)-U^*(1)}{2(1-y)}
    +\Delta^*(\psi(y))
    =
    \Delta^*(\psi(y))
    -
    \frac{1}{2(1-y)}
    \int_y^1\Delta^*(\psi(t))\,dt.
\end{align*}
The normalized integral is the average spread over $[y,1]$.
Since $t\mapsto\Delta^*(\psi(t))$ is increasing,
$
    \Delta^*(\psi(y))
    \leq
    \frac{1}{1-y}
    \int_y^1\Delta^*(\psi(t))\,dt
    \leq
    \Delta^*(\psi(1)).
$
By continuity, both bounds converge to
$\Delta^*(\psi(1))$ as $y\uparrow1$, so the average spread
converges to the same value. Consequently,
\[
    \lim_{y\uparrow1}
    \frac{
        x_{2h}^*(\psi(y))-x_{2h}^*(\psi(1))
    }{1-y}
    =
    \frac12\Delta^*(\psi(1))
    >0.
\]
Since $1-y>0$, it follows that
$  x_{2h}^*(\psi(y))>x_{2h}^*(\psi(1))$
for all $y<1$ sufficiently close to one. In particular,
$y\mapsto x_{2h}^*(\psi(y))$ cannot be increasing on its entire domain
$[\underline{y},1]$.

Now suppose that the allocation rules are Lipschitz continuous.
Since $s'$ is continuous and strictly positive on the compact
set $\Theta$, $\psi$ is Lipschitz. Consequently, both
$y\mapsto\Delta^*(\psi(y))$ and
$y\mapsto x_{2h}^*(\psi(y))$ are Lipschitz.

Let $K$ be a Lipschitz constant for
$y\mapsto\Delta^*(\psi(y))$. Since
$\Delta^*(\psi(1))>0$, choose $c<1$ sufficiently close to one
that
$
    K(1-c)<\frac12\Delta^*(\psi(c)).
$
For almost every $y\in[c,1]$,
\begin{align*}
    \frac{d}{dy}x_{2h}^*(\psi(y))
    &=
    (1-y)\frac{d}{dy}\Delta^*(\psi(y))
    -\frac12\Delta^*(\psi(y))
    \leq
    K(1-c)-\frac12\Delta^*(\psi(c))
    <0.
\end{align*}
Absolute continuity then implies that
$x_{2h}^*(\psi(y))$ is strictly decreasing on $[c,1]$.

Combining the two cases, $x_{2h}^*(\psi(y))$ cannot be increasing
on its entire domain. Under Lipschitz continuity, it is strictly
decreasing on some terminal interval $[c,\bar{y}]$.
Since $s$ is strictly increasing, setting
$\theta_0=\psi(c)$ implies that $x_{2h}^*(\theta)$ is strictly
decreasing on $[\theta_0,\bar{\theta}]$.
\end{proof}

\subsection{Proof of Theorem \ref{thm:sandwich}}

\begin{proof}
    Proposition~\ref{prop:nonnegative-spread}
implies $\Delta^*(\theta)\geq0$ for every $\theta\in\Theta$.
Since 

\[
x_{2h}^*(\theta)
    =
    x_1^*(\theta)
    +
    \bigl(1-s(\theta)\bigr)\Delta^*(\theta),  \quad
  x_{2l}^*(\theta)
    =
    x_1^*(\theta)
    -
    s(\theta)\Delta^*(\theta),
\]
and that $s(\theta)\in[0,1]$, we have 
$
    x_{2l}^*(\theta)
    \leq
    x_1^*(\theta)
    \leq
    x_{2h}^*(\theta).
$
\end{proof}

\subsection{Proof of Proposition \ref{prop:boundaryconditions}}

\begin{proof}

The first claim follows immediately from  $x^*_1(\theta)-x_{2l}^*(\theta) = s(\theta)\Delta^*(\theta)$.

For the second claim, suppose that $s(\underline{\theta})>0$.
Since $\Delta^*\geq 0$ by Proposition~2, it
suffices to rule out $\Delta^*(\underline{\theta})>0$.
Suppose, toward a contradiction, that
$\Delta^*(\underline{\theta})>0$.

For a balanced mechanism with interim payoff $u$ and spread
$\Delta$, the principal's expected loss is
\[
J(u,\Delta)
=
\int_{\Theta}
\left[
    \frac12\bigl(u(\theta)-2\theta\bigr)^2
    +s(\theta)(1-s(\theta))\Delta(\theta)^2
\right]f(\theta)\,d\theta.
\]
We construct a feasible perturbation that strictly reduces
this loss.

Fix $c\in(\underline{\theta},\bar{\theta})$, to be chosen
sufficiently close to $\underline{\theta}$, and define
\[
\phi(\theta)
=
\max\left\{\frac{c-\theta}{c-\underline{\theta}},\,0\right\},
\qquad
v(\theta)
=
\int_{\theta}^{\bar{\theta}}s'(t)\phi(t)\,dt.
\]
Thus, $\phi$ is continuous and decreasing, while $v$ is
continuously differentiable and nonnegative. Both functions
vanish on $[c,\bar{\theta}]$. For sufficiently small
$\varepsilon>0$, let
\[
\Delta_{\varepsilon}(\theta)
=
\Delta^*(\theta)-\varepsilon\phi(\theta),
\qquad
u_{\varepsilon}(\theta)
=
u^*(\theta)+\varepsilon v(\theta).
\]

We first verify feasibility. Because $\Delta^*$ and $-\phi$
are increasing, $\Delta_{\varepsilon}$ is increasing.
Moreover,
\[
u_{\varepsilon}'(\theta)
=
s'(\theta)\Delta^*(\theta)
-\varepsilon s'(\theta)\phi(\theta)
=
s'(\theta)\Delta_{\varepsilon}(\theta),
\]
so the envelope condition is preserved. Since $v\geq 0$,
we also have $u_{\varepsilon}\geq u^*\geq 0$, preserving
individual rationality. The associated balanced allocations,
\begin{align*}
    x_{1,\varepsilon}(\theta)
    =\frac{u_{\varepsilon}(\theta)}2, \; \;
    x_{2h,\varepsilon}(\theta)
    =\frac{u_{\varepsilon}(\theta)}2
      +(1-s(\theta))\Delta_{\varepsilon}(\theta), \; \;
    x_{2l,\varepsilon}(\theta)
    =\frac{u_{\varepsilon}(\theta)}2
      -s(\theta)\Delta_{\varepsilon}(\theta),
\end{align*}
are continuous and, by Lemma~1, incentive compatible.
Hence the perturbed mechanism is implementable. It coincides
with the original mechanism on $[c,\bar{\theta}]$.

The directional derivative of expected loss is
\begin{align*}
&\left.
\frac{d}{d\varepsilon}
J(u_{\varepsilon},\Delta_{\varepsilon})
\right|_{\varepsilon=0^+}
 =
\int_{\underline{\theta}}^c
\left[
    (u^*(\theta)-2\theta)v(\theta)
    -2s(\theta)(1-s(\theta))
       \Delta^*(\theta)\phi(\theta)
\right]f(\theta)\,d\theta\\
& \qquad \qquad =
\int_{\underline{\theta}}^c \phi(t)
\left[
    s'(t)\int_{\underline{\theta}}^t
    (u^*(r)-2r)f(r)\,dr
    -2s(t)(1-s(t))\Delta^*(t)f(t)
\right]dt,
\end{align*}
where the last equality follows by changing the order of
integration.

The expression in square brackets in the last line is
continuous in $t$. At $t=\underline{\theta}$, it equals
\[
-2s(\underline{\theta})(1-s(\underline{\theta}))
\Delta^*(\underline{\theta})f(\underline{\theta})<0.
\]
Indeed, $s(\underline{\theta})>0$,
$\Delta^*(\underline{\theta})>0$, and
$f(\underline{\theta})>0$, while
$s(\underline{\theta})<1$ follows from the strict increase
of $s$. Consequently, choosing $c$ sufficiently close to
$\underline{\theta}$ makes the expression in square brackets
strictly negative throughout $[\underline{\theta},c]$.
Since $\phi(t)>0$ for $t\in[\underline{\theta},c)$,
the directional derivative is strictly negative.
For sufficiently small $\varepsilon>0$, the perturbed
mechanism therefore has strictly lower expected loss,
contradicting optimality.

It follows that $\Delta^*(\underline{\theta})=0$.
Evaluating the two identities at the beginning of the proof
at $\underline{\theta}$ then gives
\[
x_1^*(\underline{\theta})
=x_{2l}^*(\underline{\theta})
=x_{2h}^*(\underline{\theta}),
\]
which proves the second claim.
\end{proof}

\subsection{Proof of Corollary \ref{cor:nonmonotone}}
\begin{proof}
By Theorem~\ref{thm:general-monotonicity} (iii), $x^*_{2h}$ is not increasing on $\Theta$. It therefore suffices to show that $x^*_{2h}$ is not decreasing on $\Theta$.

Suppose, toward a contradiction, that $x^*_{2h}$ is decreasing on $\Theta$. Since $s(\underline{\theta}) > 0$, Proposition~\ref{prop:boundaryconditions} gives $x^*_{2h}(\underline{\theta}) = x^*_1(\underline{\theta})$. Fix any $\theta \in \Theta$. The sandwich property (Theorem~\ref{thm:sandwich}) gives $x^*_{2h}(\theta) \ge x^*_1(\theta)$, and Theorem~\ref{thm:general-monotonicity} (i) gives $x^*_1(\theta) \ge x^*_1(\underline{\theta})$. Combining these with the assumed monotonicity of $x^*_{2h}$,
\[
x^*_{2h}(\underline{\theta})
\;\ge\; x^*_{2h}(\theta)
\;\ge\; x^*_1(\theta)
\;\ge\; x^*_1(\underline{\theta})
\;=\; x^*_{2h}(\underline{\theta}).
\]
Hence every inequality holds with equality, and in particular $x^*_1(\theta) = x^*_1(\underline{\theta})$ for every $\theta \in \Theta$. Thus $x^*_1$ is constant on $\Theta$.

By the balancing condition, $x^*_1 = u^*/2$, so $u^*$ is constant as well. By Lemma~\ref{lem:IC} (ii),
\[
0 \;=\; u^*(\theta) - u^*(\underline{\theta}) \;=\; \int_{\underline{\theta}}^{\theta} s'(t)\,\Delta^*(t)\,dt
\qquad \text{for every } \theta \in \Theta.
\]
The integrand is continuous and nonnegative, since $s' > 0$ and $\Delta^* \ge 0$ by Proposition~\ref{prop:nonnegative-spread}. It follows that $\Delta^* \equiv 0$ on $\Theta$. This contradicts $\Delta^* \not\equiv 0$, which is established in the proof of Theorem~\ref{thm:general-monotonicity}. Therefore $x^*_{2h}$ is not decreasing on $\Theta$, and the corollary follows.
\end{proof}

\bibliographystyle{econ} 
\bibliography{references}

\end{appendices}
\end{document}